\documentclass[aps,pre,reprint,superscriptaddress,longbibliography]{revtex4-2}
\usepackage{lmodern}
\usepackage[T1]{fontenc}
\usepackage[utf8]{inputenc}
\usepackage{mathrsfs}
\usepackage{amsmath}
\usepackage{amsthm}
\usepackage{amssymb}
\usepackage{graphicx}
\usepackage{latexsym}
\usepackage{bm}
\usepackage{amsfonts}
\usepackage{babel}
\usepackage{tikz}
\usepackage{enumitem}
\usepackage[normalem]{ulem}

\usepackage{mathtools}

\usepackage{array}

\usepackage{pgfplots}
\pgfplotsset{compat=1.7}

\allowdisplaybreaks

\usepackage{hyperref}
\hypersetup{
    breaklinks = true,
    colorlinks = true,
    citecolor = {blue},
    urlcolor = {blue},
    linkcolor = {blue}
}

\newtheorem{lemma}{Lemma}
\newtheorem{theorem}{Theorem}

\newcommand{\yfu}[1]{{#1}}

\begin{document}

\title{
Energy-Momentum-Conserving Stochastic Differential Equations and Algorithms for Nonlinear Landau-Fokker-Planck Equation
}
\author{Yichen Fu}
\email{fu9@llnl.gov}
\date{\today}
\affiliation{
Lawrence Livermore National Laboratory
}

\author{Justin R. Angus}
\affiliation{
Lawrence Livermore National Laboratory
}

\author{Hong Qin}
\affiliation{Princeton Plasma Physics Laboratory and Department of Astrophysical Sciences, Princeton University}

\author{Vasily I. Geyko}
\affiliation{
Lawrence Livermore National Laboratory
}

\begin{abstract}
	Coulomb collision is a fundamental diffusion process in plasmas that can be described by the Landau-Fokker-Planck (LFP) equation or the stochastic differential equation (SDE). While energy and momentum are conserved exactly in the LFP equation, they are conserved only on average by the conventional corresponding SDEs, suggesting that the underlying stochastic process may not be well-defined by such SDEs. In this study, we derive new SDEs with exact energy-momentum conservation for the Coulomb collision by factorizing the collective effect of field particles into individual particles and enforcing Newton's third law. These SDEs, when interpreted in the Stratonovich sense, have a particularly simple form that represents pure diffusion between particles without drag. 
    \yfu{To demonstrate that the new SDEs correspond to the LFP equation, we develop numerical algorithms that converge to the SDEs and preserve discrete conservation laws. Simulation results are presented in benchmark of various relaxation processes.} 
\end{abstract}

\maketitle

\section{Introduction}

Collisional processes in fully ionized, weakly coupled plasmas are governed by cumulative small-angle Coulomb scattering. These collisions act as diffusion in the plasma's velocity space, described by the Landau-Fokker-Planck (LFP) equation \cite{rosenbluth1957fokker,helander2005collisional}. Coulomb collisions contribute to many important processes such as the classical \cite{braginskii1965transport}, neoclassical \cite{hinton1976theory}, and turbulent \cite{lin1999effects} transport in magnetized plasmas and collisional absorption of laser radiation in plasma \cite{dawson1962high,decker1994nonlinear}. Accurate understanding and calculation of Coulomb collisions are thus important in simulating these physical processes.

Stochastic differential equations (SDEs) are differential equations with stochastic terms, making them natural models for random physical processes like Brownian motion \cite{van1992stochastic}, and are closely connected with the Fokker-Planck (FP) equations. A physical quantity described by an SDE has its distribution function governed by a corresponding FP equation, and vice versa \cite{Kloeden1992numerical,risken1996fokker}. Thus, SDEs-based methods \cite{manheimer1997langevin,cohen2010time,cadjan1999langevin,lemons2009small,zheng2021issde,dimits2013higher,rosin2014multilevel,lu2024highorder} have been frequently utilized in particle simulations to represent Coulomb collisions. However, conventional SDEs corresponding to the LFP equation conserve energy and momentum only on average. This limitation implies that the underlying stochastic process of Coulomb collisions may not be well-defined by those SDEs. As such, numerical algorithms based on them cannot naturally hold conservation laws. 
\yfu{The breakdown of conservation laws in the collision operators may lead to erroneous numerical energy and momentum transport in spatially inhomogeneous plasma. In the worst case, noise-induced numerical cooling or heating instabilities \cite{lemons1995noise} may occur and degrade the simulation quality. Therefore, \textit{ad hoc} correctors need to be applied to ensure long-time accuracy \cite{cadjan1999langevin,lemons2009small}.} 
This disadvantage has become a major obstacle in understanding Coulomb collisions and developing SDE-based algorithms. 

In this paper, we construct new SDEs corresponding to the nonlinear LFP equation for Coulomb collisions while conserving energy and momentum exactly. The SDEs describe collisions within the same species (intra-species) and between different species (inter-species). In addition, in the Stratonovich sense, these new SDEs are particularly simple as pure stochastic diffusion without deterministic drag. \yfu{In order to show that the new SDEs correspond to the LFP equation, we also develop numerical algorithms that converge to the new SDEs and hold the same conservation laws exactly. They are benchmarked in simulations of intra-species temperature isotropization and inter-species temperature relaxation processes.} 
These results provide a new description of the underlying stochastic process of Coulomb collisions and the associated structure-preserving algorithms \cite{morrison2017structure,zhang2020simulating,fu2022explicitly,hirvijoki2021structure,zonta2022multispecies}.

\section{Fokker-Planck equations and SDE}

Consider the LFP equation describing collisions between two homogeneous plasma species $\alpha$ and $\beta$, which can be either the same or different species. Let $f_\alpha$ be the distribution function of species $\alpha$. Its time evolution due to Coulomb collision with species $\beta$ is given by \cite{rosenbluth1957fokker,helander2005collisional,villani1998new}
\begin{equation}
	\left(\dfrac{\partial f_{\alpha}}{\partial t}\right)_{\alpha\beta} = 
	- \dfrac{\partial}{\partial \mathbf{v}}\cdot 
	\left[
		\mu_{\alpha\beta} f_{\alpha}
		- \dfrac{1}{2} \dfrac{\partial}{\partial \mathbf{v}}\cdot 
		(D_{\alpha\beta} f_{\alpha})
	\right],
	\label{eq:FP}
\end{equation}
where the diffusion and drag coefficients $D_{\alpha\beta}(\mathbf{v})$ and $\mu_{\alpha\beta}(\mathbf{v})$ between the two species are
\begin{subequations}
    \label{eq:diffusion_drag}
    \begin{align}
    	D_{\alpha\beta}
    	& \doteq \dfrac{L_{\alpha\beta}}{m_{\alpha}^2}
    	\int a(\mathbf{v}-\mathbf{v}')f_{\beta}(\mathbf{v}') \,\mathrm{d}\mathbf{v}'.
    	\label{eq:diffusion} \\ 
    	\mu_{\alpha\beta}
    	& \doteq \dfrac{L_{\alpha\beta}}{2 m_{\alpha}}\left(\dfrac{1}{m_{\alpha}}+\dfrac{1}{m_{\beta}}\right)
    	\int b(\mathbf{v}-\mathbf{v}') f_{\beta}(\mathbf{v}') \,\mathrm{d}\mathbf{v}',
    	\label{eq:drag}
    \end{align}
\end{subequations}
Here, $L_{\alpha\beta} \doteq ({e_\alpha^2 e_{\beta}^2}/{4\pi \epsilon_0^2}) \ln\Lambda_{\alpha\beta}$; $e_\alpha, e_\beta$ are the charges of particles;  $m_\alpha, m_\beta$ are the masses; $\ln\Lambda_{\alpha\beta}$ is the Coulomb logarithm. Function $a(\mathbf{v})$ is defined as 
\begin{equation}
	a(\mathbf{v})\doteq 
	|\mathbf{v}|^{2+\gamma} \,\Pi(\mathbf{v}), \quad
	\Pi(\mathbf{v})\doteq
	\mathrm{I}_3-\dfrac{\mathbf{vv}^\intercal}{|\mathbf{v}|^2},
	\label{eq:projector}
\end{equation}
where $\mathbf{v}=(v_1,v_2,v_3)^\intercal$, $|\mathbf{v}|=\sqrt{v_1^2+v_2^2+v_3^2}$, $\mathrm{I}_3$ is the identity matrix in 3-D, and $\Pi$ is a $3\times 3$ symmetric idempotent matrix satisfying $\Pi^2 = \Pi = (\Pi)^\intercal$. In addition, $b(\mathbf{v})\doteq \partial_\mathbf{v}\cdot a(\mathbf{v})= -2|\mathbf{v}|^\gamma \mathbf{v}$. When $\gamma=-3$, Eq.~(\ref{eq:FP}) describes Coulomb collision where $\mu$ and $D$ can be written as derivatives of the Rosenbluth potentials \cite{rosenbluth1957fokker}. When $\gamma\in(-3,1]$, Eq.~(\ref{eq:FP}) has no obvious physical meaning but is still of mathematical interest \cite{villani1998new,villani2002review}.

For the velocity $\mathbf{v}^\alpha(t)$ of particles in species $\alpha$, the It\^{o} SDE corresponding to the Fokker-Planck equation is 
\begin{equation}
	\mathrm{d} \mathbf{v}^\alpha = \mu_{\alpha\beta}(\mathbf{v}^\alpha) \mathrm{d}t + \sqrt{D_{\alpha\beta} (\mathbf{v}^\alpha)} \, \mathrm{d} \mathbf{W},
	\label{eq:Ito_DDSDE}
\end{equation}
where $\mathbf{W}(t)$ is the 3-D Wiener process, and $\sqrt{D}$ is the square-root of $D$ satisfying $\sqrt{D}{\sqrt{D}{\,}}^\intercal = D$. Because $D$ is positive semi-definite, its square root exists in general. Because the coefficients $\mu_{\alpha\beta}$ and $D_{\alpha\beta}$ depend on velocity $\mathbf{v}^\alpha$, this diffusion process is nonlinear. In addition, since the coefficients depend on the distribution function $f_\beta(\mathbf{v})$, these SDEs are generally called distribution-dependent SDE \cite{barbu2020nonlinear,wang2018distribution}.

Equation~(\ref{eq:Ito_DDSDE}) only conserves energy and momentum on average, which means the variance of energy and momentum error can accumulate over time. Physically, it describes the interaction between the ``test particle'' $\mathbf{v}^\alpha$ and the collective effect of the ``field particles'' $f_\beta(\mathbf{v})$ \cite{cohen1950electrical}, so conservation laws cannot be enforced on each particle. To overcome this difficulty, SDEs need to include interactions between individual particles instead of their collection, which we will derive next.

\subsection{Intra-species collision}

A new SDE for intra-species collisions is derived in this section. For simplicity, species labels $\alpha,\beta$ are omitted here. Consider $N$ (macro) particles represented by random variables $\mathbf{v}^i(t)$, $i\in\{1,2,...,N\}$. The distribution function $f(\mathbf{v})$ can be approximated by
\begin{equation}
	f(\mathbf{v}, t) \approx \tilde{f}(\mathbf{v}, t) \doteq w \sum_{i=1}^N \delta[\mathbf{v} - \mathbf{v}^i(t)].
	\label{eq:particle_distribution}
\end{equation}
Here, $w$ is the weight of each particle, which is assumed to be the same for all particles and is related to the number density $n$ by $w=n/N$ so that $\int \tilde{f}(\mathbf{v}) \mathrm{d}\mathbf{v}= n$.

Let $\mathbf{u}^{ij} \doteq \mathbf{v}^i - \mathbf{v}^j$ be the relative velocities, so $\mathbf{u}^{ij} = - \mathbf{u}^{ji}$. For each random variable $\mathbf{v}^i(t)$, we can use the approximated distribution function $\tilde{f}(\mathbf{v})$ in Eq.~(\ref{eq:diffusion_drag}) and write the  governing SDE (\ref{eq:Ito_DDSDE}) as 
\begin{equation}
	\mathrm{d} \mathbf{v}^i =
	\dfrac{w L}{m^2} \sum_{j,\, j\neq i}
	b(\mathbf{u}^{ij}) \mathrm{d} t
	+
	\sqrt{
	\dfrac{w L}{m^2}  \sum_{j,\, j\neq i}
	a(\mathbf{u}^{ij}) 
	} \,\mathrm{d}\mathbf{W}^i, 
	\label{eq:derivation_pre_decomposition}
\end{equation}
where the term with $j=i$ has been excluded from the summation to avoid ``self-interaction''. Strictly speaking, we also need to replace the particle weight with $w=n/(N-1)$ since the number of field particles is now $(N-1)$. At this point, the main difficulty lies in the matrix square-root operation $\sqrt{\sum_{j, \,j\neq i} a (\mathbf{u}^{ij}) }$ \cite{cadjan1999langevin}, which is not unique in general.  One method is to use Cholesky decomposition \cite{press2007numerical} to generate a lower triangular matrix as a square-root \cite{zheng2021issde,lu2024highorder,hinton2008simulating}. The Cholesky decomposition is fast, but the resulting SDEs do not conserve energy or momentum. Furthermore, the structure of $a (\mathbf{u}^{ij})$ being proportional to the symmetric idempotent matrix $\Pi(\mathbf{u}^{ij})$ is lost after the Cholesky decomposition.

In the present study, we develop a new algorithm that conserves energy and momentum. The key technique introduced is to factorize the (infinitesimal) stochastic process $\sqrt{\sum_{j, \, j\neq i} a (\mathbf{u}^{ij}) } \, \mathrm{d}\mathbf{W}^i$ instead of taking the square-root of the diffusion matrix itself. Consider three independent 3-D Gaussian random variables with mean zero and identity covariance $\mathbf{X}_1,\mathbf{X}_2,\mathbf{X}_3 \sim \mathcal{N}(0,\mathrm{I}_3)$. Let $M_1, M_2$ be two positive semi-definite matrices and define random variables $\mathbf{Y}_1 \doteq \sqrt{M_1}\,\mathbf{X}_1 + \sqrt{M_2}\,\mathbf{X}_2$ and $\mathbf{Y}_2 \doteq \sqrt{M_1 + M_2}\, \mathbf{X}_3$. One can show that the components of $\mathbf{Y}_1$ and $\mathbf{Y}_2$ have the same joint distribution, $\mathbf{Y}_1 \sim \mathbf{Y}_2 \sim \mathcal{N}(0,M_1+M_2)$, so they can be treated as the same random variable. Therefore, since the Wiener processes $\mathrm{d}\mathbf{W}^i \sim \mathcal{N}(0, \mathrm{I}_3 \mathrm{d}t)$ can be regarded as Gaussian random variables, we can factorize the stochastic process by
\begin{equation}
	\sqrt{
	\sum_{j,\, j\neq i} a (\mathbf{u}^{ij})
	} \,\mathrm{d}\mathbf{W}^i
	= \sum_{j,\, j\neq i}
	\sqrt{a (\mathbf{u}^{ij})}
	\mathrm{d}\mathbf{W}^{ij},
	\label{eq:diffusion_matrix_decomposition}
\end{equation}
where $\mathbf{W}^{ij}(t)$ and $\mathbf{W}^i(t)$ are independent 3-D Wiener processes. Furthermore, $\sqrt{a}$ can be explicitly calculated:
\begin{equation}
	\sigma(\mathbf{u}) \doteq \sqrt{a(\mathbf{u})}
	= |\mathbf{u}|^{1+\gamma/2} \Pi(\mathbf{u}).
	\label{eq:sigma_def}
\end{equation}
We can verify that $\sigma \sigma^\intercal \equiv a$ due to $ \Pi (\Pi)^\intercal = \Pi^2 = \Pi$.  The rigorous proof of Eq.~(\ref{eq:diffusion_matrix_decomposition}) is given in Appendix~\ref{app:wiener_summation} as Lemma~\ref{lemma:wiener}. 

Equation (\ref{eq:diffusion_matrix_decomposition}) switches the order of summation and matrix square-root operation at the cost of increasing the number of independent 3-D Wiener processes from $N$ to $N(N-1)$. Physically, the left-hand side (LHS) of Eq.~(\ref{eq:diffusion_matrix_decomposition}) describes an interaction from the collection of field particles, while its right-hand side (RHS) describes an interaction from each individual particle. Combining Eqs.~(\ref{eq:derivation_pre_decomposition}) and (\ref{eq:diffusion_matrix_decomposition}), we get the SDEs:
\begin{equation}
	\mathrm{d} \mathbf{v}^{i} = 
	\sum_{\substack{j=1 \\ j\neq i}}^{N}
	\left[
		\dfrac{w L}{m^2} b(\mathbf{u}^{ij}) \mathrm{d} t
		+ \dfrac{\sqrt{w L}}{m}
		{\sigma (\mathbf{u}^{ij})}
	\mathrm{d}\mathbf{W}^{ij}
	\right], 
	\label{eq:intra_particle_sys}
\end{equation}
where $\mathbf{W}^{ij}(t)$ are independent 3-D Wiener processes. Equation~(\ref{eq:intra_particle_sys}) was first proposed in Ref.~\cite{fontbona2009measurability} to approximate the LFP equation for the case of $\gamma=0$. Our results are established for $\gamma=-3$, the only physically meaningful case. In this approach, we enforce Newton's third law between each particle to enable exact energy and momentum conservation \cite{carrapatoso2014propagation,fournier2017from}. Notice that SDE~(\ref{eq:intra_particle_sys}) can be regarded as $\mathrm{d} \mathbf{v}^i \sim \sum_{j} \mathbf{F}^{ij}$, where $\mathbf{F}^{ij}$ is the effective force on particle $i$ from particle $j$. Newton's third law requires $\mathbf{F}^{ij} = -\mathbf{F}^{ji}$. Since $b(-\mathbf{u})=-b(\mathbf{u})$ and $\sigma(-\mathbf{u})=\sigma(\mathbf{u})$, the Wiener processes have to be antisymmetric, i.e., $\mathbf{W}^{ij}(t) = -\mathbf{W}^{ji}(t)$, hence the number of independent 3-D Wiener processes is reduced to $N(N-1)/2$.

With the antisymmetric Wiener processes, SDE~(\ref{eq:intra_particle_sys}) conserves energy and momentum exactly.  The system's momentum is $\mathbf{P} \doteq \sum_{i} w m \mathbf{v}^{i}$, whose conservation is a direct result of Newton's third law. The conservation of energy $E \doteq \sum_{i} w m |\mathbf{v}^{i}|^2/2$ can be proved using It\^{o}'s lemma for stochastic calculus \cite{Kloeden1992numerical,oksendal2003stochastic} and is presented in Appendix~\ref{app:energy_momentum_conservation}.

In addition to the exact conservation laws, the It\^{o} SDE~(\ref{eq:intra_particle_sys}) has a particularly simple form in Stratonovich sense \cite{van1981ito}. A stochastic integral $\int F(t) \mathrm{d}W$ is interpreted as $\sum_i F(t_i)[W(t_{i+1}) - W(t_i)]$ by It\^{o} and as $\sum_i F(\frac{t_{i+1}+t_i}{2})[W(t_{i+1}) - W(t_i)]$ by Stratonovich. It\^{o}'s interpretation is preferred mathematically due to its martingale property, which means that for a stochastic process written as an It\^{o} integral $X_t = \int_0^t f(X_s)\mathrm{d}W_s$, its expectation $\langle X_{t} \rangle$ is always 0 when $f$ is well-behaved. Stratonovich's interpretation benefits from obeying the usual chain rule of calculus \cite{stratonovich1966new}. Nevertheless, SDEs in the It\^{o} sense can be transformed into the Stratonovich sense with the same diffusion but modified drag. One can show, as in Appendix~\ref{app:ito_to_stratonovich}, that SDE~(\ref{eq:intra_particle_sys}) in Stratonovich sense is
\begin{equation}
	\mathrm{d} \mathbf{v}^{i} = 
	\dfrac{\sqrt{w L}}{m}
	\sum_{\substack{j=1 \\ j\neq i}}^{N}
		\sigma (\mathbf{u}^{ij}) \circ \mathrm{d}\mathbf{W}^{ij}, 
	\quad
	\mathbf{W}^{ij}=-\mathbf{W}^{ji},
	\label{eq:intra_particle_sys_stratonovich}
\end{equation}
where $\circ$ denotes Stratonovich's interpretation. The exact cancellation of the drag term indicates that the stochastic force is purely diffusive in the Stratonovich sense.

\subsection{Inter-species collision}

The SDE we derived can be generalized to inter-species collision straightforwardly. Consider $N_\alpha$ particles in species $\alpha$ and $N_\beta$ particles in species $\beta$, represented by random variables $\mathbf{v}^{\alpha,i}$ and $\mathbf{v}^{\beta,j}$, $i\in\{1,2,...,N_\alpha\}$, $j\in \{1,2,...,N_\beta\}$. Their distribution functions can be approximated by $\tilde{f}_\alpha(\mathbf{v},t) \doteq w \sum_{i=1}^{N_\alpha} \delta[\mathbf{v} - \mathbf{v}^{\alpha,i}(t)]$ and $\tilde{f}_\beta(\mathbf{v},t) \doteq w \sum_{j=1}^{N_\beta} \delta[\mathbf{v} - \mathbf{v}^{\beta,j}(t)]$. Here, the particle weight is assumed to be the same for all particles in all species, so they are related to the particle number density by $w=n_\alpha/N_\alpha=n_\beta/N_\beta$.

Let relative speeds be $\mathbf{u}^{ij} \doteq \mathbf{v}^{\alpha,i}-\mathbf{v}^{\beta,j}$, so here $\mathbf{u}^{ij} \neq -\mathbf{u}^{ji}$. We can first derive the SDEs for inter-species collision in It\^{o} sense using $\tilde{f}_\alpha$ and $\tilde{f}_\beta$ in Eq.~(\ref{eq:diffusion_drag}), factorize the stochastic processes using Eq.~(\ref{eq:diffusion_matrix_decomposition}), we get the It\^{o} SDE for inter-species collision:
\begin{widetext}
    \begin{subequations}
        \label{eq:inter_particle_sys}
        \begin{align}
        	\mathrm{d} \mathbf{v}^{\alpha,i} 
        	&= \sum_{j=1}^{N_\beta}
        	\left[
        		\dfrac{ wL_{\alpha\beta}}{2 m_{\alpha}}\left(\dfrac{1}{m_{\alpha}}+\dfrac{1}{m_{\beta}}\right)
        		b(\mathbf{u}^{ij}) \mathrm{d}t + 
        		\dfrac{\sqrt{w L_{\alpha\beta}}}{m_\alpha}
        		\sigma (\mathbf{u}^{ij}) \mathrm{d} \mathbf{W}^{ij}
        	\right], 
        	\label{eq:inter_particle_sys_1} \\
        	\mathrm{d} \mathbf{v}^{\beta,j} 
        	&= - \sum_{i=1}^{N_\alpha}
        	\left[
        		\dfrac{ wL_{\alpha\beta}}{2 m_{\beta}}\left(\dfrac{1}{m_{\alpha}}+\dfrac{1}{m_{\beta}}\right)
        		b(\mathbf{u}^{ij}) \mathrm{d}t + 
        		\dfrac{\sqrt{w L_{\alpha\beta}}}{m_\beta}
        		\sigma (\mathbf{u}^{ij}) \mathrm{d} \mathbf{W}^{ij}
        	\right],
        	\label{eq:inter_particle_sys_2}
        \end{align}
    \end{subequations}
\end{widetext}
where $\mathbf{W}^{ij}(t)$ are $N_\alpha N_\beta$ independent 3-D Wiener processes. Those SDEs can also be transformed to the Stratonovich sense, which are
\begin{subequations}
\label{eq:inter_particle_sys_stratonovich}
    \begin{align}
        	\mathrm{d} \mathbf{v}^{\alpha,i} 
    	&= \dfrac{\sqrt{w L_{\alpha\beta}}}{m_\alpha} \sum_{j=1}^{N_\beta}
    	\sigma (\mathbf{u}^{ij}) \circ \mathrm{d} \mathbf{W}^{ij}, 
    	\label{eq:inter_particle_sys_stratonovich_1} \\
    	\mathrm{d} \mathbf{v}^{\beta,j} 
    	&= - \dfrac{\sqrt{w L_{\alpha\beta}}}{m_\beta} \sum_{i=1}^{N_\alpha}
    		\sigma (\mathbf{u}^{ij}) \circ \mathrm{d} \mathbf{W}^{ij},
      \label{eq:inter_particle_sys_stratonovich_2}
    \end{align}
\end{subequations}
The minus signs in Eq.~ (\ref{eq:inter_particle_sys_2}) and (\ref{eq:inter_particle_sys_stratonovich_2}) represent Newton's third law. Similar to the inter-species case, the SDEs preserve energy and momentum exactly.

It is interesting to discuss a limiting case where $m_\alpha$ remains constant but $m_\beta \to \infty$. This is the case in the electron-ion collision and is known as the pitch-angle scattering. Since $\sigma_{\beta\alpha} \sim 1 / m_\beta \to 0$, we get $\mathrm{d}\mathbf{v}^{\beta,j} \to 0$, which means particles in species-$\beta$ barely move due to their heavy mass. So, we can approximate $\mathbf{v}^{\beta,j}$ as time-independent. Furthermore, the distribution function of heavy species is almost a single delta function if it has the same temperature as the light species. Therefore, we can replace all $\mathbf{v}^{\beta,j}$ with a single constant speed $\mathbf{v}^{\beta}$. 

Next, we consider the motion of the light species, which is described by Eq.~(\ref{eq:inter_particle_sys_stratonovich_1}). Let $\mathbf{u}^i \doteq \mathbf{v}^{\alpha,i}  - \mathbf{v}^{\beta}$, we have 
\begin{equation}
    \mathrm{d} \mathbf{v}^{\alpha,i} 
	= \sum_{j=1}^{N_\beta} 
	\sigma_{\alpha\beta} (\mathbf{u}^i) \circ \mathrm{d} \mathbf{W}^{ij} 
    = \sigma_{\alpha\beta} (\mathbf{u}^i) \circ 
     \sum_{j=1}^{N_\beta} \mathrm{d} \mathbf{W}^{ij}.
\end{equation}
Since $\mathbf{W}^{ij}$ are independent 3-D Wiener processes, we can apply Eq.~(\ref{eq:diffusion_matrix_decomposition}) and get:
\begin{equation}
    \begin{split}
    	\mathrm{d} \mathbf{v}^{\alpha,i} 
    	& = \sqrt{N_\beta} \, \sigma_{\alpha\beta} (\mathbf{u}^i) \circ 
    	\mathrm{d} \mathbf{W}^i \\
    	& = \dfrac{\sqrt{w N_\beta L_{\alpha\beta}}}{m_\alpha} \big|\mathbf{u}^{i}\big|^{1+\gamma/2} \Pi(\mathbf{u}^{i})
    	\circ \mathrm{d} \mathbf{W}^i. 
    \end{split}
\end{equation}
Notice we have $w N_\beta = n_\beta$, so in the Coulomb collision with $\gamma=-3$, the SDE becomes:
\begin{align}
	\mathrm{d} \mathbf{v}^{\alpha,i} 
     & = \sqrt{\dfrac{\bar{L}}{|\mathbf{u}^{i}|}}
	\left( \mathrm{I}_3 - \dfrac{\mathbf{u}^{i}(\mathbf{u}^{i})^\intercal}{|\mathbf{u}^{i}|^2} \right)
	\circ \mathrm{d} \mathbf{W}^i.
	\label{eq:pitch_angle}
\end{align}
where $\bar{L} \doteq ({e_\alpha^2 e_{\beta}^2 n_\beta}/{4\pi \epsilon_0^2 m_\alpha^2}) \ln\Lambda_{\alpha\beta}$. In particular, when the heavy species is motionless, $\mathbf{v}^\beta=0$, $\mathbf{u}^i = \mathbf{v}^{\alpha,i}$, the SDE~(\ref{eq:pitch_angle}) returns to the SDE for pitch-angle scattering discussed in Refs.~\cite{zhang2020simulating,fu2022explicitly}.

\section{Numerical evaluation of SDEs}

\subsection{Numerical algorithms}

\yfu{We now look for numerical algorithms that converge to SDEs and preserve the conservation laws in discrete time, so that we can numerically demonstrate the correspondence between the new SDEs are the LFP equation.} SDEs in the Stratonovich sense can be calculated naturally using the midpoint scheme \cite{milstein2002numerical,milstein2004stochastic}. However, since the fully implicit midpoint scheme has to be solved by iterative methods such as the fixed-point iteration \cite{hoffman2018numerical}, it may fail to converge when the time step is not small enough \cite{zonta2022multispecies}. Here, we propose a modified midpoint scheme that is explicitly solvable and conserves energy-momentum exactly at finite time steps. 

We discretize the time domain $t\in[0,T]$ to $\{t_k\}$ with $t_{k+1} = t_k + \Delta t$ and the time step $\Delta t$. Define $\mathbf{v}^{i}_{k} \doteq \mathbf{v}^{i}(t_{k})$ and $\mathbf{u}_k^{ij} \doteq \mathbf{v}^{i}_{k} - \mathbf{v}^{j}_{k}$. The quantities in the half step are defined as $\mathbf{v}^{i}_{k+1/2} \doteq (\mathbf{v}^{i}_{k}+\mathbf{v}^{i}_{k+1})/2$. Noticing that $\sigma (\mathbf{u}) \circ \mathrm{d}\mathbf{W} \sim (\mathbf{u}\times \mathrm{d}\mathbf{W})\times \mathbf{u}$, we use a mixture of implicit and explicit methods \cite{zhang2020simulating,fu2022explicitly} to discretize SDE~(\ref{eq:intra_particle_sys_stratonovich}). To simplify the notation in this section, we define 
\begin{equation}
	\boldsymbol{\Omega}^{ij}_k \doteq \dfrac{\mathbf{u}^{ij}_k \times \Delta \mathbf{W}^{ij} }{|\mathbf{u}^{ij}_k|^{1-\gamma/2}}.
\end{equation}
So, SDE~(\ref{eq:intra_particle_sys_stratonovich}) can be discretized as 
\begin{equation}
	\mathbf{v}^{i}_{k+1} - \mathbf{v}^{i}_{k}
	= \dfrac{\sqrt{w L}}{m}
	\sum_{\substack{j=1 \\ j\neq i}}^{N}
	\boldsymbol{\Omega}^{ij}_k \times \mathbf{u}_{k+\frac{1}{2}}^{ij},
	\label{eq:modified_midpoint_intra}
\end{equation}
where we have $\Delta \mathbf{W}^{ij} = -\Delta \mathbf{W}^{ji} \sim \mathcal{N}(0, \mathrm{I}_3 \Delta t)$. 

Similarly, for inter-species collision, SDE~(\ref{eq:inter_particle_sys_stratonovich}) can be discretized as 
\begin{equation}
    \begin{split}
        \mathbf{v}^{\alpha,i}_{k+1} - \mathbf{v}^{\alpha,i}_{k}
        &= \dfrac{\sqrt{w L_{\alpha\beta}}}{m_\alpha}
        \sum_{j=1}^{N_\beta}
        \boldsymbol{\Omega}^{ij}_k \times \mathbf{u}_{k+\frac{1}{2}}^{ij}, \\
        \mathbf{v}^{\beta,j}_{k+1} - \mathbf{v}^{\beta,j}_{k}
        &= - \dfrac{\sqrt{w L_{\alpha\beta}}}{m_\beta} 
        \sum_{i=1}^{N_\alpha}
        \boldsymbol{\Omega}^{ij}_k \times \mathbf{u}_{k+\frac{1}{2}}^{ij}.
    \end{split}
    \label{eq:modified_midpoint_inter}
\end{equation}
It is straightforward to prove that the algorithm holds conservation laws exactly. The momentum conservation is a direct result of Newton's third law, while the energy conservation is shown in Appendix~\ref{app:numerical_conservation}. Since the RHS of Eq.~(\ref{eq:modified_midpoint_intra}) and (\ref{eq:modified_midpoint_inter}) depends on $\mathbf{v}_{k+1}$ linearly, they can be solved explicitly, which is presented in Appendix~\ref{app:explicit_solution_to_modified_midpoint}. In addition, as shown in Appendix~\ref{app:convergence_of_modified_midpoint}, this algorithm converges to SDE~(\ref{eq:intra_particle_sys_stratonovich}) and (\ref{eq:inter_particle_sys_stratonovich}) with strong order $1/2$.

\subsection{Benchmark with relaxation processes}

Here, we benchmark the numerical solutions of SDEs with analytical results in relaxation processes for both intra- and inter-species Coulomb collision with $\gamma=-3$. 

The first case is the isotropization of temperatures along different directions within one species. Assume the plasma is Maxwellian parallel and perpendicular to $z$-axis with temperature $T_\parallel$ and $T_\perp$. Coulomb collision will drive temperature in two directions towards equilibrium, which is governed by \cite{ichimaru1970relaxation,richardson20192019}: 
\begin{equation}
	\dfrac{\mathrm{d} T_\perp}{\mathrm{d} t} = 
	-\dfrac{1}{2} \dfrac{\mathrm{d} T_\parallel}{\mathrm{d} t}
	= \tau_\mathrm{iso}^{-1} (T_\parallel - T_\perp).
\end{equation}
Let $A\doteq T_\perp/T_\parallel - 1>0$, the isotropization time $\tau_\mathrm{iso}$ is
\begin{equation}
	\tau_\mathrm{iso}^{-1} \doteq 
	\dfrac{e^4 n \ln\Lambda }{8\pi^{3/2}\epsilon_0^2 \sqrt{m}\, T_\parallel^{3/2}}
	A^{-2} f(A), 
\end{equation}
where $f(A)=(A+3)\arctan(\sqrt{A})/\sqrt{A}-3$. A numerical benchmark is presented in Fig.~\ref{fig:relaxations}(a), (b) with initial temperature $T_\perp, T_\parallel=4, 1$. The process is simulated by $2^8$ particles and time step $\Delta t = 10^{-2} \tau_\mathrm{iso,0}$, where $\tau_\mathrm{iso,0}$ is the initial isotropization time. For convenience, $m$, $n$, $\epsilon_0$, $\ln\Lambda$ were set to 1. This process is simulated in $2^{11}$ ensembles where the initial distribution and time evolution are sampled independently. The averaged temperature is shown in Fig.~\ref{fig:relaxations}(a) in a log-time scale, which coincides with the analytical solution. The average absolute error in energy and momentum in Fig.~\ref{fig:relaxations}(b) is around machine precision, with the momentum error being larger due to the initial momentum being small.

The second case is the relaxation of different isotropic temperatures between two species. Assuming two species are Maxwellian with temperatures $T_\alpha$ and $T_\beta$, their relaxation due to Coulomb collisions is given by \cite{wesson2011tokamaks}
\begin{equation}
	\dfrac{\mathrm{d}T_\alpha}{\mathrm{d} t} = \tau_{\alpha\beta}^{-1} ( T_\beta - T_\alpha ),
\end{equation}
where the relaxation time $\tau_{\alpha\beta}$ is
\begin{equation}
	\tau_{\alpha\beta}^{-1} \doteq
	\dfrac{e_\alpha^2 e_\beta^2 n_\beta \ln\Lambda_{\alpha\beta} }{3\sqrt{2} \pi^{3/2} \epsilon_0^2 m_\alpha m_\beta}
	\left(\dfrac{T_\alpha}{m_\alpha} + \dfrac{T_\beta}{m_\beta} \right)^{-3/2}.
\end{equation} 
In the numerical example shown in Fig.~\ref{fig:relaxations}(c),(d), we set temperatures $T_1, T_2=4, 1$, masses $m_1, m_2=1,5$, charges $e_1, e_2=2,-1$, densities $n_1, n_2 = 1, 2$, and other constants as 1. Species 1 and 2 are represented by $2^6$ and $2^7$ particles proportional to their density. The time step is chosen to be $\Delta t = 10^{-3} \tau_{11,0}$, where $\tau_{11,0}$ is the initial relaxation time in species 1. The process is, again, simulated in $2^{11}$ ensembles where both inter- and intra-species collision are turned on. The average temperature in Fig.~\ref{fig:relaxations}(c) equilibrates slightly slower than the analytical solution, which may be because the time step is not small enough. The energy and momentum conservation is also around machine precision as shown in Fig.~\ref{fig:relaxations}(d).
\begin{figure}[ht]
	\centering 
	\includegraphics[width=0.95\columnwidth]{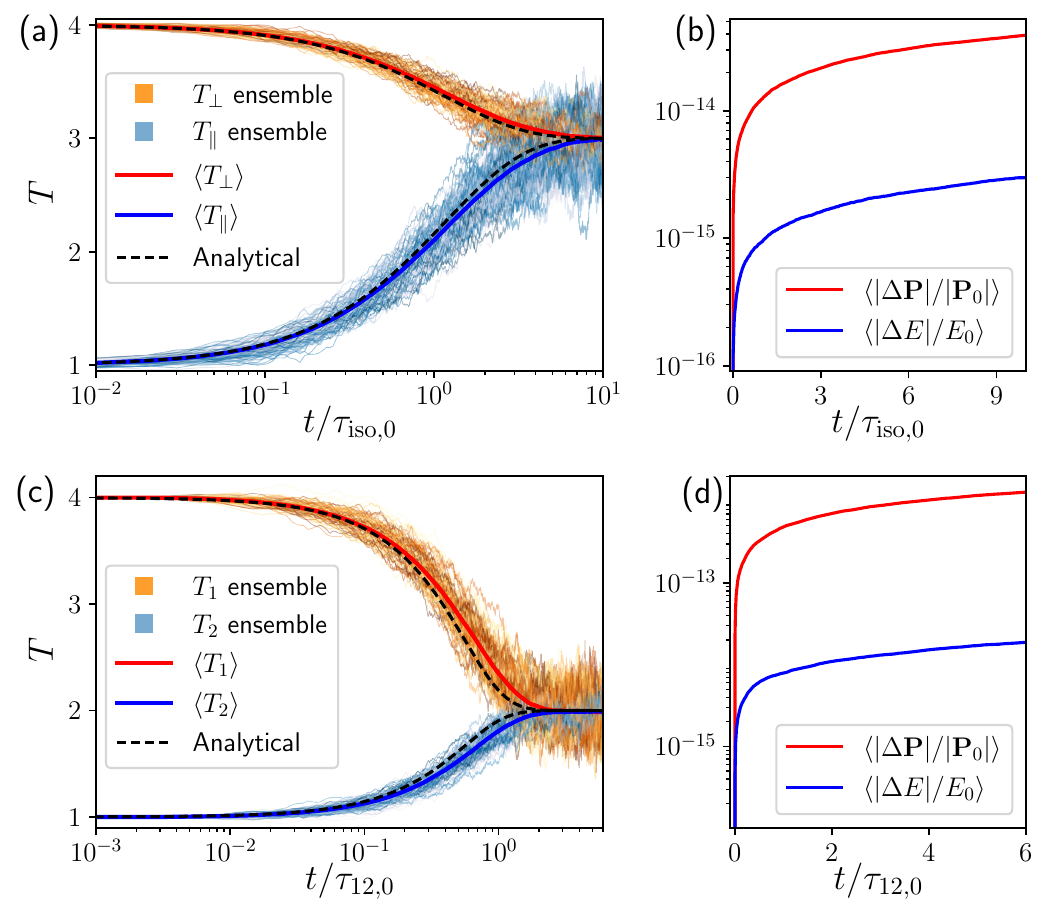} 
	\caption{Benchmark of relaxation processes. (a),(b): intra-species temperature isotropization. (c),(d): inter-species temperature relaxation. (a), (c): the first 100 sample patches are shown with a color gradient, and the average temperature and analytical solution are shown in solid and dashed lines. (b), (d) are average absolute errors in energy and momentum.}
	\label{fig:relaxations}
\end{figure}

Overall, the two cases present good agreement between 3-D simulations based on SDEs and the analytical theory. Therefore, SDEs~(\ref{eq:intra_particle_sys_stratonovich})-(\ref{eq:inter_particle_sys_stratonovich}) indeed approximate Coulomb collisions in plasmas as an energy-momentum-conserving stochastic process.

\section{Discussion}

This paper develops new SDEs~(\ref{eq:intra_particle_sys_stratonovich})-(\ref{eq:inter_particle_sys_stratonovich}) that correspond to the LFP equation~(\ref{eq:FP}) for both intra- and inter-species collisions. The new SDEs preserve energy and momentum exactly and, therefore, provide a better description of the underlying stochastic process for Coulomb collisions. Numerical algorithms that preserve the conservation laws exactly in discrete time are constructed and benchmarked with analytical solutions for intra-species temperature isotropization and inter-species temperature relaxation. 

The SDEs derived for the LFP equation hold exact conservation laws only for equally weighted particles. In certain applications, unequally weighted particles are preferable due to the scale separation of density or cell volume in the system \cite{miller1994coulomb,nanbu1998weighted,higginson2020corrected,angus2024moment}. Further modification will be necessary to extend the current SDE to unequal weight. Additionally, Coulomb collisions in gyrokinetics can also be described by the FP equation \cite{brizard2004guiding,hirvijoki2013monte}, so SDEs may also be derived utilizing the decomposition~(\ref{eq:diffusion_matrix_decomposition}) and enforcement of Newton's third law. However, more investigation is needed to study whether the resulting SDE preserves conservation laws or whether it is pure diffusion in the Stratonovich sense.

\section{Acknowledgment}
This research was supported by the US Department of Energy through Contracts DE-AC52-07NA27344 (LLNL), DE-AC02-09CH1146 (PPPL), and LLNL-JRNL-870465. This work was also supported by the LLNL-LDRD Program under Project No. 23-ERD-007.

\appendix
\section{Summation of Wiener processes \label{app:wiener_summation}}

In this section, we prove a lemma regarding the summation of Wiener processes, which is useful in deriving SDEs. Consider $N$ constant $d\times d$ positive semi-definite matrices $M_i$, $i\in\{1,2,...,N\}$ and $N+1$ independent $d$-dimensional Wiener processes $\mathbf{W}_i(t)$, $i\in\{0,1,...,N\}$. Each Wiener process is a column vector
\begin{equation}
	\mathbf{W}_i(t) \doteq [W_{i,1}(t),W_{i,2}(t), \cdots, W_{i,d}(t)]^\intercal, 
\end{equation}
where $W_{i,j}(t) - W_{i,j}(0) \sim \mathcal{N}(0,t)$ and $\mathcal{N}(0,t)$ indicate normal distribution with mean $0$ and variance $t$. For any matrix $M$, $\sqrt{M}$ is defined as a matrix decomposition satisfying $\sqrt{M} {\sqrt{M}\,}^\intercal=M$, which exists if $M$ is positive semi-definite.

Now, define two random processes $\mathbf{X}_1(t)$ and $\mathbf{X}_2(t)$ as 
\begin{subequations}
\begin{align}
	\mathbf{X}_1(t) &\doteq \sum_{i=1}^N \sqrt{M_i} \,\mathbf{W}_i(t), \\
	\mathbf{X}_2(t) &\doteq \sqrt{\textstyle \sum_{i=1}^N M_i} \, \mathbf{W}_0(t).
\end{align}
\end{subequations}
It is straightforward to see that each component in $\mathbf{X}_1(t)$ and $\mathbf{X}_2(t)$ is the summation of Wiener processes with different variances. Furthermore, it can be proved that $\mathbf{X}_1(t)$ and $\mathbf{X}_2(t)$ can be regarded as the same random process in the following sense. 
At any $t>0$, let the joint distribution of $d$ components in $\mathbf{X}_i(t)$ be $P_i(x_{1},x_{2},...,x_{d},t)$. We have:
\begin{lemma}
	Random processes $\mathbf{X}_1(t)$ and $\mathbf{X}_2(t)$ are the same process in the sense that at any time $t>0$, the components in $\mathbf{X}_1(t)$ and components in $\mathbf{X}_2(t)$ have the same joint distribution, i.e., $P_1(x_{1},x_{2},...,x_{d},t)=P_2(x_{1},x_{2},...,x_{d},t)$.
	\label{lemma:wiener}
\end{lemma}

\begin{proof}
	Because each component of $\mathbf{X}_1(t)$ and $\mathbf{X}_2(t)$ are the summation of Wiener processes, at a given time $t$, their joint distribution must be multivariate Gaussian distributions with mean 0. Therefore, it is enough to prove that the components of $\mathbf{X}_1(t)$ and $\mathbf{X}_2(t)$ have the same covariance matrix.

	The covariance matrix of $\mathbf{X}_1(t)$ can be calculated as follows. Let $\langle \cdot \rangle$ denote the expectation of a random variable, we have
     \begin{equation}
        \begin{split}
            \langle \mathbf{X}_1(t) \mathbf{X}_1^\intercal(t) \rangle
            &= \left\langle \sum_{i,j} \sqrt{M_i} \,\mathbf{W}_i(t) \mathbf{W}_j^\intercal (t) {\sqrt{M_j}\,}^\intercal \right\rangle \\
            &= \sum_{i,j} \sqrt{M_i} \, \left\langle \mathbf{W}_i(t) \mathbf{W}_j^\intercal (t)\right\rangle {\sqrt{M_j}\,}^\intercal
        \end{split}
     \end{equation}
	Because $\mathbf{W}_i(t)$ are independent $d$-dimensional Wiener processes, we have
	\begin{equation}
		\left\langle \mathbf{W}_i(t) \mathbf{W}_j^\intercal (t)\right\rangle_{\mu\nu}
		= \left\langle {W}_{i,\mu}(t) {W}_{j,\nu} (t)\right\rangle = t \,\delta_{ij} \delta_{\mu\nu},
	\end{equation}
    where $\mu, \nu \in \{ 1, 2, ..., d \}$. So the covariance matrix becomes
	\begin{equation}
		\langle \mathbf{X}_1(t) \mathbf{X}_1^\intercal(t)\rangle = 
		\sum_{i,j} \sqrt{M_i} \, \delta_{ij} \mathrm{I}_d {\sqrt{M_j}\,}^\intercal
		=\sum_{i=1}^N M_i,
	\end{equation}
	where $\mathrm{I}_d$ is a $d\times d$ identity matrix. On the other hand, the covariance matrix of $\mathbf{X}_2(t)$ is
	\begin{subequations}
         \begin{align}
            & \langle\mathbf{X}_2(t) \mathbf{X}_2^\intercal(t)\rangle\\
            = & \sqrt{\textstyle \sum_{i} M_i } \, \langle \mathbf{W}_0(t) \mathbf{W}_0^\intercal (t)\rangle {\sqrt{\textstyle \sum_{i} M_i }\,}^\intercal \\
            = & \sqrt{\textstyle \sum_{i} M_i } \, \mathrm{I}_d {\sqrt{\textstyle \sum_{i} M_i }\,}^\intercal
            = \sum_{i=1}^N M_i .
          \end{align}
	\end{subequations}
	Therefore, the covariance matrices of $\mathbf{X}_1(t)$ and $\mathbf{X}_2(t)$ are the same. 
\end{proof}

Lemma~\ref{lemma:wiener} can be regarded as the matrix extension of the following summation rule of Gaussian random variables. Let $X, Y, Z \sim \mathcal{N}(0,1)$ be independent Gaussian random variables with mean 0 and variance 1. For two constants $a,b > 0$, we have $\sqrt{a}\,X + \sqrt{b}\,Y \sim \sqrt{a+b}\, Z \sim \mathcal{N}(0, a+b)$. Lemma~\ref{lemma:wiener} can also be applied to SDEs so that we have:
\begin{equation}
	\sum_{i=1}^N \sqrt{M_i} \,\mathrm{d} \mathbf{W}_i(t) = 
	\sqrt{\textstyle \sum_{i=1}^N M_i } \, \mathrm{d}\mathbf{W}_0(t).
\end{equation}

\section{Properties of SDEs}

This section presents rigorous proof of various properties of SDEs discussed in the main article. For simplicity, we only present the proof for the inter-species collision. The proof for the intra-species case will be an analogy to the inter-species case with equal mass. We will first rewrite the SDE to better label its indices, then prove its conservation law, transform it to Stratonovich form, and finally discuss its singularities.

\subsection{SDE and its indices}

The SDE for inter-species collision is presented in Eq.~(\ref{eq:inter_particle_sys}) whose indices are slightly different from the standard ones. For example, let $X(t)$ be a $d$-dimensional random process, a standard SDE will be written as 
\begin{equation}
	\mathrm{d}X^i = e^i \mathrm{d}t + \sum_{j=1}^N F^{ij} \mathrm{d}W^j, \quad i\in\{1,2,...,d\}.
	\label{eq:standard_SDE}
\end{equation}
In our SDE~(\ref{eq:inter_particle_sys}), the index $i$ is associated with species $\alpha$, and the index $j$ is associated with species $\beta$. However, when calculation gets complicated, such a convention can easily be confused. In addition, the component label $n$ is omitted since we use vector notation $\mathbf{v}$ instead of $v_n$, which could also cause trouble. To make the calculations in the next few sections easier, we introduce the following labeling of SDE~(\ref{eq:inter_particle_sys}).

Firstly, the SDE~(\ref{eq:inter_particle_sys}) includes random variables ${v}^{s,i}_n$, where the species index $s\in\{\alpha, \beta\}$, particle index $i\in\{1,2,...,N_s\}$, and the component index $n\in\{1,2,3\}$. So, the dimension $d$ in Eq.~(\ref{eq:standard_SDE}) is $3(N_\alpha+N_\beta)$ and the index $i$ in Eq.~(\ref{eq:standard_SDE}) should be replaced by three indices $(s,i,n)$. 

Next, the SDE~(\ref{eq:inter_particle_sys}) has independent Wiener processes $W^{jk}_m(t)$, where the particle index $j\in\{1,2,...,N_\alpha\}$, $k\in\{1,2,...,N_\beta\}$, and the component index $m\in\{1,2,3\}$. So the number of independent Wiener processes $N$ in Eq.~(\ref{eq:standard_SDE}) is $3N_\alpha N_\beta$ and the index $j$ in Eq.~(\ref{eq:standard_SDE}) should be replaced by another three indices $(j,k,m)$. 

With the above definition, coefficient $F^{ij}$ in Eq.~(\ref{eq:standard_SDE}) is a $3(N_\alpha+N_\beta) \times 3N_\alpha N_\beta$ matrix, which should be labeled by 6 indices $(s,i,n;j,k,m)$. We can separate the two component indices $m,n \in \{1,2,3\}$ and write the coefficient as $F^{s,i;j,k}_{nm}$. So, the SDE~(\ref{eq:inter_particle_sys_1}) and (\ref{eq:inter_particle_sys_2}) can be rewritten as follows:
\begin{equation}
	\mathrm{d}{v}^{s,i}_n = e^{s,i}_n \, \mathrm{d}t 
	+ \sum_{j=1}^{N_\alpha} \sum_{k=1}^{N_\beta} \sum_{m=1}^3
	F^{s,i;j,k}_{nm} \mathrm{d}W^{jk}_m,
	\label{eq:standard_SDE_new_index}
\end{equation}
where $s\in\{\alpha, \beta\}$, $i\in\{1,2,...,N_s\}$, and $n\in\{1,2,3\}$. 

For the drag coefficients $e^{s,i}_n$, we have
\begin{subequations}
	\begin{align}
		e^{\alpha,i}_n 
		&= \sum_{j=1}^{N_\beta} 
		\left[
			\dfrac{ wL_{\alpha\beta}}{2 m_{\alpha}}\left(\dfrac{1}{m_{\alpha}}+\dfrac{1}{m_{\beta}}\right)
			b(\mathbf{u}^{ij})
		\right]_n, \\  
		e^{\beta,i}_n
		&= - \sum_{j=1}^{N_\alpha} 
		\left[ 
			\dfrac{ wL_{\alpha\beta}}{2 m_{\beta}}\left(\dfrac{1}{m_{\alpha}}+\dfrac{1}{m_{\beta}}\right)
			b(\mathbf{u}^{ji})
		\right]_n. 
	\end{align}
\end{subequations}
Notice that both $e^{\alpha,i}_n$ and $e^{\beta,i}_n$ are summations over index $j$ but have different summation limits.

For any random variable ${v}^{s,i}_n$, it only interacts with $3N_{\bar{s}\neq s}$ among all $3N_\alpha N_\beta$ Wiener processes. Therefore, the diffusion coefficient $F^{s,i;j,k}_{nm}$ is sparse and non-zero only when $\{s=\alpha, i=j\}$ or $\{s=\beta, i=k\}$. The non-zero components of $F^{s,i;j,k}_{nm}$ are
\begin{subequations}
    \label{eq:F_non_zero}
    \begin{align}
    	F^{\alpha,j;j,k}_{nm} &= 
    	\left[
    		\dfrac{\sqrt{w L_{\alpha\beta}}}{m_\alpha}
    		\sigma (\mathbf{u}^{jk})
    	\right]_{nm}, \\  
    	F^{\beta,k;j,k}_{nm} &= - 
    	\left[ 
    		\dfrac{\sqrt{w L_{\alpha\beta}}}{m_\beta}
    		\sigma (\mathbf{u}^{jk})
    	\right]_{nm}.
    \end{align}
\end{subequations}

To further simplify the notation, we will use reduced mass $m_\mathrm{r} = m_\alpha m_\beta / (m_\alpha + m_\beta)$ and define:
\begin{equation}
	b^{ij}_n \doteq \left[b(\mathbf{u}^{ji})\right]_n, \quad 
	\sigma^{ij}_{mn} \doteq \left[ \sigma (\mathbf{u}^{ij}) \right]_{nm},
\end{equation}
which will be used in the rest of the Appendix.

\subsection{Energy- and momentum-conservation \label{app:energy_momentum_conservation}}

This section proves conservation laws in SDE~(\ref{eq:inter_particle_sys}). Using indices in Eq.~(\ref{eq:standard_SDE_new_index}), the energy and momentum can be written as
\begin{subequations}
    \label{eq:energy_and_momentum_definition}
    \begin{align}
    	E &\doteq \sum_{s}^{\{\alpha, \beta\}} \,
    	\sum_{i=1}^{N_s} \sum_{n=1}^3
    	\dfrac{1}{2} w m_s (v^{s,i}_n)^2, \\  
    	P_n &\doteq 
    	\sum_{s}^{\{\alpha, \beta\}} \,
    	\sum_{i=1}^{N_s} w m_s v^{s,i}_n,
    \end{align}
\end{subequations}
where $s\in\{\alpha, \beta\}$ is the particle species and $w$ is the particle weight. The momentum conservation is straightforward since it directly results from Newton's third law. 

To evaluate the time evolution of energy, we need to use It\^{o}'s lemma, which is briefly summarized as follows. Consider a $d$-dimensional random process $X(t)=[X^1(t),...,X^d(t)]$ governed by SDE~(\ref{eq:standard_SDE}). Let $U=U(x^1, ..., x^d)$ be a smooth \yfu{scalar} function, one can define a 1-D random process $Y(t) = U[X^1(t), ..., X^d(t)]$. The It\^{o}'s lemma states that the SDE for $Y$ is \cite{Kloeden1992numerical}:
\begin{equation*}
    \begin{split}
    	\mathrm{d} Y 
    	= & \left(
    		\sum_{i=1}^d e^i \dfrac{\partial U}{\partial x^i}
    		+ \dfrac{1}{2} \sum_{j=1}^N \, \sum_{i,k=1}^d F^{ij} F^{kj} \dfrac{\partial^2 U}{\partial x^i \partial x^k}
    	\right) \mathrm{d}t \\
    	& \quad + \sum_{j=1}^N \sum_{i=1}^d F^{ij} \dfrac{\partial U}{\partial x^i} \mathrm{d}W^j.
    \end{split}    
\end{equation*}

In the case of energy conservation, we can define \yfu{scalar} function $U=E$ in Eq.~(\ref{eq:energy_and_momentum_definition}), whose derivatives are
\begin{subequations}
    \begin{align}
    	&\dfrac{\partial U}{\partial v^{s,i}_n} 
        = w m_s v^{s,i}_n, \\ 
    	&\dfrac{\partial^2 U}{\partial v^{s,i}_n \partial v^{\bar{s},\bar{i}}_{\bar{n}}} 
        = w m_s \delta_{s\bar{s}} \,\delta_{i\bar{i}} \,\delta_{n\bar{n}}.  
    \end{align}
\end{subequations}
Using the indices Eq.~(\ref{eq:standard_SDE_new_index}), the It\^{o}'s lemma tells us the SDE for kinetic energy $E$ is
\begin{subequations}
	\label{eq:ito_energy_mid_step}
    \begin{align}
    	\mathrm{d}E = &  
    	w\sum_{s,i,n} e^{s,i}_n m_s v^{s,i}_n \mathrm{d}t 
        \label{eq:ito_energy_mid_step_1}\\
    	& + \dfrac{w}{2} 
    	\sum_{j,k,m} \, \sum_{s,i,n}
    	F^{s,i;j,k}_{nm} F^{s,i;j,k}_{nm} 
    	m_s \mathrm{d}t 
        \label{eq:ito_energy_mid_step_2} \\
    	& + w\sum_{s,i,n} \, \sum_{j,k,m}
    	F^{s,i;j,k}_{nm} m_s v^{s,i}_n \mathrm{d}W^{jk}_m.
    	\label{eq:ito_energy_mid_step_3}
    \end{align}
\end{subequations}
Here, we denote the three terms in Eq.~(\ref{eq:ito_energy_mid_step_1})-(\ref{eq:ito_energy_mid_step_3}) as $T_1$, $T_2$, and $T_3$, which will be evaluated individually. 

The first term $T_1$ is 
\begin{subequations}
	\begin{align}
		& T_1 / w \doteq
		\sum_{s,i,n} e^{s,i}_n m_s v^{s,i}_n \\
		 = & 
		\sum_{i=1}^{N_\alpha} \sum_{n=1}^3  e^{\alpha,i}_n m_\alpha v^{\alpha,i}_n 
		+ \sum_{i=1}^{N_\beta} \sum_{n=1}^3  e^{\beta,i}_n m_\beta v^{\beta,i}_n \\ 
		= & 
		\sum_{i=1}^{N_\alpha}\sum_{j=1}^{N_\beta} \sum_{n=1}^3 
			\dfrac{ wL_{\alpha\beta}}{2 m_{\alpha} m_\mathrm{r}}
			m_\alpha b^{ij}_n v^{\alpha,i}_n \nonumber \\
		& - \sum_{i=1}^{N_\beta} \sum_{j=1}^{N_\alpha} \sum_{n=1}^3
			\dfrac{ wL_{\alpha\beta}}{2 m_{\beta} m_\mathrm{r}}
			m_\beta b^{ij}_n v^{\beta,i}_n \\ 
		= & \sum_{i=1}^{N_\alpha}\sum_{j=1}^{N_\beta} \sum_{n=1}^3 
			\dfrac{ wL_{\alpha\beta}}{2 m_\mathrm{r}}
			b^{ij}_n
		\left( v^{\alpha,i}_n - v^{\beta,j}_n \right) 
		\\
		= & \sum_{i,j,n}
		 \dfrac{ wL_{\alpha\beta}}{2 m_\mathrm{r}} b^{ij}_n {u}_n^{ij} 
		= - \sum_{i,j} \dfrac{wL_{\alpha\beta}}{m_\mathrm{r}}
		\big|\mathbf{u}^{ij}\big|^{2+\gamma},
	\end{align}
\end{subequations}
where we have used $b(\mathbf{u})\doteq \partial_\mathbf{u}\cdot a(\mathbf{u})= -2|\mathbf{u}|^\gamma \mathbf{u}$.

Since the components of $F$ is non-zero only for terms in Eq.~(\ref{eq:F_non_zero}), we can calculate $T_2$ in Eq.~(\ref{eq:ito_energy_mid_step_2}) as
\begin{subequations}
	\begin{align}
		& T_2 / w \doteq
		\dfrac{1}{2}
		\sum_{m,n} \sum_{j,k} \sum_{s,i} F^{s,i;j,k}_{nm} F^{s,i;j,k}_{nm}  m_s \\
		= & \dfrac{1}{2} 
			\sum_{m,n} \sum_{j,k} 
			\Big( F^{\alpha,j;j,k}_{nm} F^{\alpha,j;j,k}_{nm} m_\alpha \nonumber \\
		  & \qquad \qquad \qquad +  F^{\beta,k;j,k}_{nm} F^{\beta,k;j,k}_{nm} m_\beta \Big) \\ 
		= & \dfrac{1}{2} \sum_{m,n} \sum_{j,k}
			\Bigg( \dfrac{{w L_{\alpha\beta}}}{m_\alpha}
			\sigma^{jk}_{nm} \sigma^{jk}_{nm} \nonumber \\ 
		  & \qquad \qquad \qquad + \dfrac{{w L_{\alpha\beta}}}{m_\beta}
			\sigma^{jk}_{nm} \sigma^{jk}_{nm} \Bigg).
	\end{align}
\end{subequations}
Notice that matrix $\sigma(\mathbf{u})$ is defined in Eq.~(\ref{eq:sigma_def}), so we have
\begin{subequations}
	\begin{align}
		\sum_{m,n} [\sigma(\mathbf{u})]_{mn}[\sigma(\mathbf{u})]_{mn}
		= \mathrm{Tr}\Big\{ [\sigma(\mathbf{u})] [\sigma(\mathbf{u})]^\intercal \Big\} \\
		= \mathrm{Tr}\Big[ a(\mathbf{u}) \Big]
		= |\mathbf{u}|^{2+\gamma} \mathrm{Tr}\Big[ \Pi(\mathbf{u}) \Big]
		= 2 |\mathbf{u}|^{2+\gamma}.
	\end{align}
\end{subequations}
Therefore, $T_2$ can be written as 
\begin{align}
	T_2 / w = \sum_{j,k}
	\dfrac{w L_{\alpha\beta}}{m_\mathrm{r}} 
	\big|\mathbf{u}^{jk}\big|^{2+\gamma},
\end{align}
which means $T_1 + T_2=0$ and the deterministic term in Eq.~(\ref{eq:ito_energy_mid_step}) vanishes.  

We can calculate the third term $T_3$ in Eq.~(\ref{eq:ito_energy_mid_step_3}) as:
\begin{subequations}
	\begin{align}
		& T_3/w \doteq 
		\sum_{j,k}\sum_{m,n}\sum_{s,i}
		F^{s,i;j,k}_{nm} m_s v^{s,i}_n \mathrm{d}W^{jk}_m \\ 
		= & \sum_{j,k}\sum_{m,n}
			\Big( F^{\alpha,j;j,k}_{nm} m_\alpha v^{\alpha,j}_n \mathrm{d}W^{jk}_m \nonumber \\[-3pt]
		  & \qquad \qquad + F^{\beta,k;j,k}_{nm} m_\beta v^{\beta,k}_n \mathrm{d}W^{jk}_m \Big) \\[3pt]
		= & \sum_{j,k}\sum_{m,n}
		\sqrt{w L_{\alpha\beta}} 
		\left(
		\sigma^{jk}_{nm} v^{\alpha,j}_n 
		- 
		\sigma^{jk}_{nm} v^{\beta,k}_n
		\right)
		\mathrm{d}W^{jk}_m \\ 
		= & \sum_{j,k}\sum_{m,n}
		\sqrt{w L_{\alpha\beta}} 
		\sigma^{jk}_{nm} 
		u^{jk}_n
		\mathrm{d}W^{jk}_m.
	\end{align}
\end{subequations}
Since $\sigma(\mathbf{u})$ is proportional to projector $\Pi(\mathbf{u})$, which satisfy $\Pi(\mathbf{u})\mathbf{u}=0$, we can see the stochastic term $T_3=0$. 

In summary, with the help of It\^{o}'s lemma, we have proved that the energy satisfies $\mathrm{d}E = 0$, which means it is exactly conserved in the SDE.

\subsection{Transformation from It\^{o} to Stratonovich \label{app:ito_to_stratonovich}}

Next, we transform SDE~(\ref{eq:inter_particle_sys_1}) and (\ref{eq:inter_particle_sys_2}) into Stratonovich sense. Let $X(t)$ be a $d$-dimensional random process, the following SDEs in It\^{o} and Stratonovich sense are equivalent:
\begin{subequations}
	\begin{align}
		\text{It\^{o}:}& \quad 
		\mathrm{d}X^i = e^i \mathrm{d}t + \sum_{j=1}^N F^{ij} \mathrm{d}W^j,
		\label{eq:ito_SDE} \\ 
		\text{Stratonovich:}& \quad 
		\mathrm{d}X^i = \bar{e\,}^i \mathrm{d}t + \sum_{j=1}^N F^{ij} \circ \mathrm{d}W^j,
	\end{align}
\end{subequations}
where $i\in\{1,2,...,d\}$. Their diffusion coefficients $F$ are same, but drag coefficients $e$ and $\bar{e\,}$ are related by
\begin{equation}
	e^i - \bar{e\,}^i 
	= \dfrac{1}{2} \sum_{j=1}^d \sum_{k=1}^N F^{jk} \dfrac{\partial F^{ik}}{\partial x_j}.
	\label{eq:ito_to_stratonovich}
\end{equation}

Using the indices defined in Eq.~(\ref{eq:standard_SDE_new_index}), the transformation factor $\Delta^{s,i}_n \doteq e^{s,i}_n - \bar{e\,}^{s,i}_n$ is:
\begin{equation}
	\Delta^{s,i}_n =
	\dfrac{1}{2} \sum_{\bar{s}}^{\{\alpha,\beta\}} \sum_{\bar{i}=1}^{N_{\bar{s}}} 
	\sum_{j=1}^{N_\alpha} \sum_{k=1}^{N_\beta} \sum_{m,n=1}^3 \,
	F^{\bar{s},\bar{i};j,k}_{\bar{n} m} \dfrac{\partial F^{s,i;j,k}_{nm}}{\partial v^{\bar{s},\bar{i}}_{\bar{n}}}.
\end{equation}
At given indices $j$ and $k$, the non-zero components of $F$ is given by Eq.~(\ref{eq:F_non_zero}). So, firstly, the summation over indices $\bar{s}$ and $\bar{i}$ can be replaced by
\begin{align}
	\Delta^{s,i}_n = 
	& \dfrac{1}{2} \sum_{j,k} \sum_{m,\bar{n}}
	F^{\alpha,j;j,k}_{\bar{n} m} \dfrac{\partial F^{s,i;j,k}_{nm}}{\partial v^{\alpha,j}_{\bar{n}}} 
	\nonumber \\ 
	& + \dfrac{1}{2} \sum_{j,k} \sum_{m,\bar{n}}
	F^{\beta,k;j,k}_{\bar{n} m} \dfrac{\partial F^{s,i;j,k}_{nm}}{\partial v^{\beta,k}_{\bar{n}}}.
\end{align}
Next, consider the case where $s=\alpha$, so $F^{s,i;j,k}_{nm}$ is non-zero only when $i=j$. In this case, the summation of the index $j$ can be replaced by:
\begin{subequations}
	\begin{align}
		\Delta^{s,i}_n = 
		& \dfrac{1}{2} \sum_{k=1}^{N_\beta} \sum_{m,\bar{n}}
		F^{\alpha,i;i,k}_{\bar{n} m} \dfrac{\partial F^{\alpha,i;i,k}_{nm}}{\partial v^{\alpha,i}_{\bar{n}}} 
		\nonumber \\
		& + \dfrac{1}{2} \sum_{k=1}^{N_\beta} \sum_{m,\bar{n}}
		F^{\beta,k;i,k}_{\bar{n} m} \dfrac{\partial F^{\alpha,i;i,k}_{nm}}{\partial v^{\beta,k}_{\bar{n}}} 
		\\ 
		= & \dfrac{1}{2} 
		\sum_{k=1}^{N_\beta} \sum_{m,\bar{n}}
		\dfrac{{w L_{\alpha\beta}}}{m_\alpha}
		\Bigg( 
		\dfrac{1}{m_\alpha}
		\sigma^{ik}_{\bar{n}m}
		\dfrac{\partial \sigma^{ik}_{nm}}{\partial v^{\alpha,i}_{\bar{n}}} 
		\nonumber \\ 
		& \qquad \qquad \qquad \qquad - 
		\dfrac{1}{m_\beta}
		\sigma^{ik}_{\bar{n}m}
		\dfrac{\partial \sigma^{ik}_{nm}}{\partial v^{\beta,k}_{\bar{n}}} 
		\Bigg).
		\label{eq:ito_to_stratonovich_tmp_1}
	\end{align}
\end{subequations}

Notice that $\mathbf{u}^{ik} = \mathbf{v}^{\alpha,i} - \mathbf{v}^{\beta,k}$, so for any function $f(\mathbf{u}^{ik})$, we have
\begin{equation}
	\dfrac{\partial f(\mathbf{u}^{ik})}{\partial \mathbf{u}^{ik}}
	= \dfrac{\partial f(\mathbf{u}^{ik})}{\partial \mathbf{v}^{\alpha,i}}
	= - \dfrac{\partial f(\mathbf{u}^{ik})}{\partial \mathbf{v}^{\beta,k}}. 
\end{equation}
Thus, Eq.~(\ref{eq:ito_to_stratonovich_tmp_1}) becomes:
\begin{equation}
	\Delta^{s,i}_n
	= \dfrac{w L_{\alpha\beta}}{2 m_\alpha m_\mathrm{r}} 
	\sum_{k=1}^{N_\beta} \sum_{m,\bar{n}}
	\sigma^{ik}_{\bar{n} m} 
	\dfrac{\partial \sigma^{ik}_{nm}}{\partial u^{ik}_{\bar{n}}}.
\end{equation}
We can verify that 
\begin{equation}
	\sum_{m,\bar{n}=1}^3 \left[ {\sigma}(\mathbf{u}) \right]_{\bar{n}m}
	\dfrac{\partial \left[ {\sigma}(\mathbf{u}) \right]_{nm}}{\partial u_{\bar{n}}}
	= -2 |\mathbf{u}|^\gamma u_n
	= b(\mathbf{u}).
\end{equation}
So, finally, we get 
\begin{equation}
	\Delta^{s,i}_n
	= \dfrac{w L_{\alpha\beta}}{2 m_\alpha m_\mathrm{r}}
	\sum_{k=1}^{N_\beta} b_n^{ik}
	= e^{\alpha,i}_n, 
\end{equation}
which indicates that $\bar{e\,}^{\alpha,i}_n = 0$. Similarly, we can also show that $\bar{e\,}^{\beta,i}_n = 0$, which means that the drag term in Stratonovich SDE vanishes.

\subsection{Singularities}

Despite the distinctions with the standard SDE in Eq.~(\ref{eq:Ito_DDSDE}), it has been rigorously proved that the SDE~(\ref{eq:intra_particle_sys}) also converges to the LFP equation describing the Coulomb collision. With $N$ random variables $\mathbf{v}^i(t)$ satisfying SDE~(\ref{eq:intra_particle_sys}), one can define the Wasserstein distance $\mathcal{W}_2(\tilde{f}, f)$ between the approximated distribution $\tilde{f}(\mathbf{v}) \doteq w \sum_{i=1}^N \delta(\mathbf{v} - \mathbf{v}^i)$ and the exact distribution $f(\mathbf{v})$ from the LFP equation (\ref{eq:FP}). Roughly speaking, Ref.~\cite{fournier2017from} has proved that when $\gamma\in[0,1]$, the Wasserstein distance $\mathcal{W}_2(\tilde{f}, f) \sim \mathcal{O}(N^{-1/3})$ for any $0<t<\infty$. In other words, $\tilde{f}(\mathbf{v})$ converges to the LFP equation's solution when the number of particles $N$ is large enough. 

It is necessary to mention that the rigorous convergence theorem is only proved for $\gamma\in[0,1]$, which does not include the case of Coulomb collision with $\gamma=-3$. This is mainly due to the mathematical difficulty in treating the singularities in functions $a(\mathbf{v})$ and $b(\mathbf{v})$ at $\mathbf{v}\to 0$ when $\gamma < 0$. For example, $b(\mathbf{v}) \sim |\mathbf{v}|^\gamma \mathbf{v}$, so its derivatives at $\mathbf{v}\to0$ diverge if $\gamma<0$. On the other hand, despite the singularities, the diffusion and drag coefficients defined in Eq.~(\ref{eq:diffusion_drag}) remain finite. Thus, we may argue that the difficulties due to the singularities are mathematical instead of physical. Therefore, we can still use SDEs~(\ref{eq:intra_particle_sys_stratonovich})-(\ref{eq:inter_particle_sys_stratonovich}) to study Coulomb collision for $\gamma=-3$.

Nevertheless, although the singularities should not affect the physics of Coulomb collision, they can still cause trouble during the numerical integration of the SDEs. For example, the SDE~(\ref{eq:intra_particle_sys_stratonovich}) is written in Stratonovich sense, so one may want to solve it by the midpoint method \cite{milstein2002numerical,milstein2004stochastic}, which can be written as 
\begin{equation}
	\mathbf{v}^{i}_{k+1} - \mathbf{v}^{i}_{k}
	= \sum_{\substack{j=1 \\ j\neq i}}^{N}
	\sigma (\mathbf{u}^{ij}_{k+1/2}) \, \Delta \mathbf{W}^{ij},
	\label{eq:midpoint_intra}
\end{equation}
where $\mathbf{v}^{i}_{k} \doteq \mathbf{v}^{i}(t_{k})$, $\mathbf{v}^{i}_{k+1/2} \doteq (\mathbf{v}^{i}_{k}+\mathbf{v}^{i}_{k+1})/2$, and $\mathbf{u}^{ij}_{k+1/2}\doteq \mathbf{v}^{i}_{k+1/2} - \mathbf{v}^{j}_{k+1/2}$. Here, $\Delta \mathbf{W}^{ij}\doteq \mathbf{W}^{ij}(t_{k+1}) - \mathbf{W}^{ij}(t_{k}) \sim \mathcal{N}(0, \mathrm{I}_3 \Delta t)$ are 3-D Gaussian random variables with mean 0 and covariance matrix $\mathrm{I}_3 \Delta t$. However, such an algorithm may not be practically useful. The main difficulty is due to the implicitness of the midpoint method, which can only be solved using iterative methods such as the fixed-point iteration. On the other hand, the singularity of $\sigma(\mathbf{u})$ means it could be extremely large when $\mathbf{u}$ is small, which can cause the failure of iterative methods. 

One numerical example is presented in Fig.~\ref{fig:failure_of_midpoint} to show the temperature isotropization within the same species. SDE~(\ref{eq:intra_particle_sys_stratonovich}) is calculated using the midpoint method~(\ref{eq:midpoint_intra}) with a fixed point iteration, where the initial guess of $\mathbf{v}_{k+1}$ in each step is chosen to be the result of the Euler-Maruyama method of Eq.~(\ref{eq:intra_particle_sys}). The number of particles is $N=128$ and the time step is chosen to be $\Delta t = 10^{-4} \tau_\mathrm{iso,0}$, which is already extremely small. 

\begin{figure}[ht]
	\centering 
	\includegraphics[width=0.95\columnwidth]{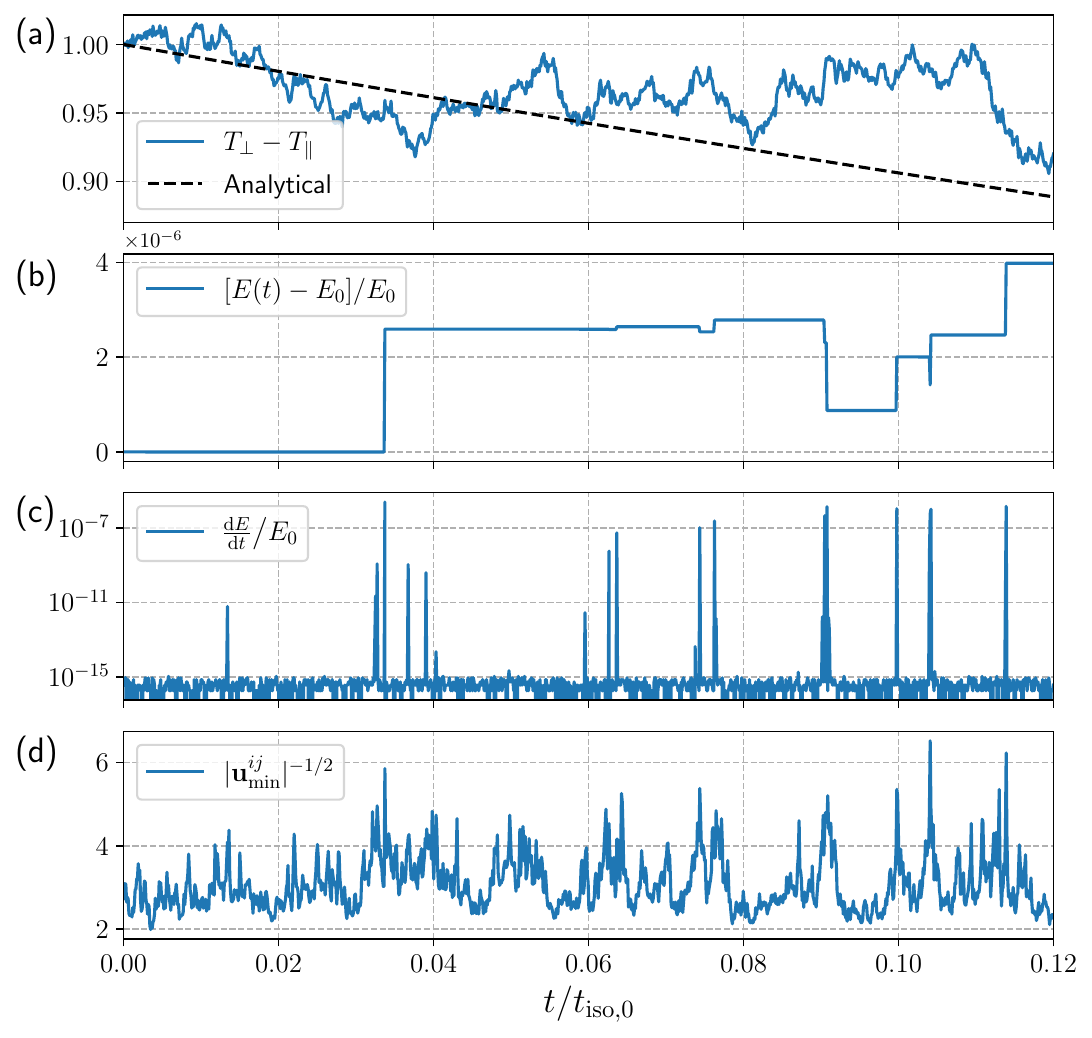} 
	\caption{Demonstration of the unconvergence in the midpoint method. (a) The temperature difference in one sample path compared to the analytical solution. (b) The evolution of total energy. (c) The time derivative of total energy. (d) $|\mathbf{u}_\mathrm{min}^{ij}|^{-1/2}$ where $\mathbf{u}_\mathrm{min}^{ij}$ is the minimum value of relative velocities in each time.}
	\label{fig:failure_of_midpoint}
\end{figure}

The temperature difference in one sample path is presented in Fig.~\ref{fig:failure_of_midpoint}(a) while the total energy of the system is shown in Fig.~\ref{fig:failure_of_midpoint}(b). We can see that the total energy remains constant most of the time but occasionally jumps by a finite amount. This is because the fixed-point iteration fails to converge at some time step, which is also observed on other non-SDE-based implicit numerical integrators, such as in Ref.~\cite{zonta2022multispecies}. Indeed, the derivative of total energy, shown in Fig.~\ref{fig:failure_of_midpoint}(c), is almost at machine precision most of the time except for several abrupt jumps. Because the divergent term in $\sigma(\mathbf{u})$ is proportional to $|\mathbf{u}|^{-1/2}$, we calculate the minimum value of relative velocities $\mathbf{u}_\mathrm{min}^{ij}$ in each time step and plot $|\mathbf{u}_\mathrm{min}^{ij}|^{-1/2}$ in Fig.~\ref{fig:failure_of_midpoint}(d). We can see that \yfu{there is a strong correlation} between the non-zero derivatives of energy and the peaks of $|\mathbf{u}_\mathrm{min}^{ij}|^{-1/2}$, though the correspondence is not one-to-one. Thus, we believe that the singularities in $\sigma(\mathbf{u})$ are the root cause of the unconvergence of the iterative methods. Larger time steps, such as $\Delta t = 10^{-2} \tau_\mathrm{iso,0}$ used in Fig.~\ref{fig:relaxations}(a) in the main article, will make the unconvergence issue much worse.

On the other hand, the modified midpoint method we derived in Eq.~(\ref{eq:modified_midpoint_intra}) is explicitly solvable, which avoids the convergence issue in iterative methods. However, it still suffers from the singularity because when $\sigma(\mathbf{u})$ is large, one still needs small time steps to ensure numerical accuracy in the modified midpoint method. Such an issue has been discussed in section~\ref{sec:complexity_reduction}
.

\section{Properties of the numerical method}

This section proves various properties of the modified midpoint method both analytically and numerically. Similar to the previous Appendix, we only demonstrate the properties of inter-species collision. We will first prove this method conserves energy and momentum exactly and then construct its explicit solution. Then, we will prove such a numerical method indeed converges to the SDE we're interested in. Finally, we will show the evolution of the distribution functions in our numerical examples.

\subsection{Energy- and momentum-conservation \label{app:numerical_conservation}}

This section proves the conservation law in the numerical method. The change of momentum in each step is
\begin{subequations}
	\begin{align}
		& \Delta \mathbf{P}_{k+1} \doteq \sum_{s}^{\{\alpha,\beta\}} \sum_{i=1}^{N_s}
		w m_s \Delta \mathbf{v}^{s,i}_{k+1} \\ 
		\sim & \sum_{i,j}
		\left(
			\boldsymbol{\Omega}^{ij}_k \times \mathbf{u}_{k+\frac{1}{2}}^{ij} 
			- \boldsymbol{\Omega}^{ij}_k \times \mathbf{u}_{k+\frac{1}{2}}^{ij}
		\right)
		= 0.
		\label{eq:numerical_momentum_conservation}
	\end{align}
\end{subequations}
We can see, again, that momentum conservation is the direct result of Newton's third law. 

The change of energy in each step is 
\begin{subequations}
	\begin{align}
		& \Delta E_k \doteq \sum_{s}^{\{\alpha,\beta\}} \sum_{i=1}^{N_s}
		\dfrac{1}{2} w m_s \left[(\mathbf{v}^{s,i}_{k+1})^2 - (\mathbf{v}^{s,i}_k)^2\right] \\ 
		= & \sum_{s}^{\{\alpha,\beta\}} \sum_{i=1}^{N_s} w m_s 
		\mathbf{v}^{s,i}_{k+\frac{1}{2}} \cdot \Delta \mathbf{v}^{s,i}_{k+1} \\ 
		\sim & \sum_{i,j}
		\Bigg[
			\mathbf{v}^{\alpha,i}_{k+\frac{1}{2}} \cdot
			\left(
				\boldsymbol{\Omega}^{ij}_k \times \mathbf{u}_{k+\frac{1}{2}}^{ij}
			\right) 
			\nonumber \\[-2pt]
			& \qquad \quad - \mathbf{v}^{\beta,j}_{k+\frac{1}{2}} \cdot
			\left(
				\boldsymbol{\Omega}^{ij}_k \times \mathbf{u}_{k+\frac{1}{2}}^{ij}
			\right)
		\Bigg] \\ 
		\sim & \sum_{i,j}
		\mathbf{u}_{k+\frac{1}{2}}^{ij}
		\cdot
		\left(
			\boldsymbol{\Omega}^{ij}_k \times \mathbf{u}_{k+\frac{1}{2}}^{ij}
		\right) = 0.
		\label{eq:numerical_energy_conservation}
	\end{align}
\end{subequations}
Therefore, the modified midpoint method conserves energy and momentum exactly in each time step.

\subsection{Explicit solution \label{app:explicit_solution_to_modified_midpoint}}

This section explicitly solves the velocities at the $k+1$ step in the modified midpoint scheme. For any constant vector $\mathbf{A}=(A_1,A_2,A_3)^\intercal$ , we can define a $3\times 3$ matrix $\hat{\mathbf{A}}$ as 
\begin{equation}
	\hat{\mathbf{A}} \doteq 
	\begin{pmatrix}
		0 & -A_3 & A_2 \\ 
		A_3 & 0 & -A_1 \\ 
		-A_2 & A_1 & 0
	\end{pmatrix}.
	\label{eq:hat_map}
\end{equation}
For another vector $\mathbf{B}=(B_1,B_2,B_3)^\intercal$, we have
\begin{equation}
	\mathbf{A}\times \mathbf{B} = \hat{\mathbf{A}} \mathbf{B}.
\end{equation}

Thus, we can rewrite Eq.~(\ref{eq:modified_midpoint_inter}) as 
\begin{widetext}
	\begin{align}
		\left( 
			\mathrm{I}_3 
			- \dfrac{\sqrt{w L_{\alpha\beta}}}{2 m_\alpha} \sum_{j=1}^{N_\beta} \hat{\boldsymbol{\Omega}}^{ij}_k 
		\right) 
		\mathbf{v}^{\alpha, i}_{k+1} 
		+ \dfrac{\sqrt{w L_{\alpha\beta}}}{2 m_\alpha} \sum_{j=1}^{N_\beta} \hat{\boldsymbol{\Omega}}^{ij}_k \mathbf{v}^{\beta,j}_{k+1} 
		= &
		\left( \mathrm{I}_3 + \dfrac{\sqrt{w L_{\alpha\beta}}}{2 m_\alpha} \sum_{j=1}^{N_\beta} \hat{\boldsymbol{\Omega}}^{ij}_k\right) 
		\mathbf{v}^{\alpha,i}_{k}
		- \dfrac{\sqrt{w L_{\alpha\beta}}}{2 m_\alpha} \sum_{j=1}^{N_\beta} \hat{\boldsymbol{\Omega}}^{ij}_k \mathbf{v}^{\beta,j}_{k},
		\label{eq:tmp_explicit_solve1}\\ 
		\left( \mathrm{I}_3 - \dfrac{\sqrt{w L_{\alpha\beta}}}{2 m_\beta}\, \sum_{i=1}^{N_\alpha} \hat{\boldsymbol{\Omega}}^{ij}_k \right)
		\mathbf{v}^{\beta, j}_{k+1} 
		+ \dfrac{\sqrt{w L_{\alpha\beta}}}{2 m_\beta}\, \sum_{i=1}^{N_\alpha} \hat{\boldsymbol{\Omega}}^{ij}_k \mathbf{v}^{\alpha,i}_{k+1}
		= &
		\left( \mathrm{I}_3 + \dfrac{\sqrt{w L_{\alpha\beta}}}{2 m_\beta}\, \sum_{i=1}^{N_\alpha} \hat{\boldsymbol{\Omega}}^{ij}_k\right) 
		\mathbf{v}^{\beta,j}_{k}
		- \dfrac{\sqrt{w L_{\alpha\beta}}}{2 m_\beta}\, \sum_{i=1}^{N_\alpha} \hat{\boldsymbol{\Omega}}^{ij}_k \mathbf{v}^{\alpha,i}_{k}.
		\label{eq:tmp_explicit_solve2}
	\end{align}
\end{widetext}

Let $\mathbf{v}_k^s = (\mathbf{v}_k^{s,1},\mathbf{v}_k^{s,2},...,\mathbf{v}_k^{s,N})^\intercal$ and $\mathbf{v}_{k+1}^s = (\mathbf{v}_{k+1}^{s,1},\mathbf{v}_{k+1}^{s,2},...,\mathbf{v}_{k+1}^{s,N})^\intercal$ with $s\in\{\alpha,\beta\}$. Define $N_\alpha \times N_\alpha$ diagonal matrix $\mathbf{M}_1=M_1^{ij}$ and $N_\alpha \times N_\beta$ matrix $\mathbf{M}_2=M_2^{ij}$ as:
\begin{subequations}
	\begin{align}
		M_1^{ii} & \doteq - \dfrac{\sqrt{w L_{\alpha\beta}}}{2 m_\alpha}  \sum_{j=1}^{N_\beta} \hat{\boldsymbol{\Omega}}^{ij}_k, 
		\quad i\in\{1,2,...,N_\alpha \};  \\ 
		M_2^{ij} & \doteq \dfrac{\sqrt{w L_{\alpha\beta}}}{2 m_\alpha} \hat{\boldsymbol{\Omega}}^{ij}_k, 
		\quad i\neq j.
	\end{align}
\end{subequations}
All other elements of $\mathbf{M}_1$ and $\mathbf{M}_2$ are zero. Similarly, we can define $N_\beta \times N_\beta$ diagonal matrix $\mathbf{N}_1=N_1^{ij}$ and $N_\alpha \times N_\beta$ matrix $\mathbf{N}_2=N_2^{ij}$ as:
\begin{subequations}
	\begin{align}
		N_1^{jj} & \doteq - \dfrac{\sqrt{w L_{\alpha\beta}}}{2 m_\beta} \sum_{i=1}^{N_\alpha} \hat{\boldsymbol{\Omega}}^{ij}_k, 
		\quad j\in\{1,2,...,N_\beta \};  \\ 
		N_2^{ij} & \doteq \dfrac{\sqrt{w L_{\alpha\beta}}}{2 m_\beta} \hat{\boldsymbol{\Omega}}^{ij}_k,
		\quad i\neq j.
	\end{align}
\end{subequations}
We can rewrite Eq.~(\ref{eq:tmp_explicit_solve1}), (\ref{eq:tmp_explicit_solve2}) as 
\begin{align*}
	(\mathrm{I} + \mathbf{M}_1) \, \mathbf{v}^\alpha_{k+1} + \mathbf{M}_2 \, \mathbf{v}^\beta_{k+1}
	&= (\mathrm{I} - \mathbf{M}_1) \, \mathbf{v}^\alpha_{k} - \mathbf{M}_2 \, \mathbf{v}^\beta_{k},\\ 
	(\mathrm{I} + \mathbf{N}_1) \, \mathbf{v}^\beta_{k+1}  + \mathbf{N}_2^\intercal\,  \mathbf{v}^\alpha_{k+1}
	&= (\mathrm{I} - \mathbf{N}_1) \, \mathbf{v}^\beta_{k}  - \mathbf{N}_2^\intercal\,  \mathbf{v}^\alpha_{k}.
\end{align*}
Finally, we can define matrix $\mathbf{Q}$ as 
\begin{equation}
	\mathbf{Q} \doteq 
	\begin{pmatrix}
		\mathbf{M}_1 & \mathbf{M}_2 \\
		\mathbf{N}_2^\intercal & \mathbf{N}_1
	\end{pmatrix}.
\end{equation}
So, $\mathbf{v}^\alpha_{k+1}$ and $\mathbf{v}^\beta_{k+1}$ can be solved as 
\begin{equation}
	\begin{pmatrix}\mathbf{v}^\alpha_{k+1}\\[3pt] \mathbf{v}^\beta_{k+1}\end{pmatrix} 
	= (\mathrm{I} + \mathbf{Q})^{-1} (\mathrm{I} - \mathbf{Q})
	\begin{pmatrix}\mathbf{v}^\alpha_{k}\\[3pt] \mathbf{v}^\beta_{k}\end{pmatrix}.
\end{equation}
Here, $(\mathrm{I} + \mathbf{Q})^{-1} (\mathrm{I} - \mathbf{Q})$  is the Cayley transform of  $\bf{Q}$.  When $m_\alpha=m_\beta$, $\mathbf{Q}$ is antisymmetric, i.e., $\mathbf{Q}^{\intercal}=-\mathbf{Q}$, and its Cayley transform  is a rotation \cite{zhang2020simulating,fu2022explicitly,Qin2013Boris}.

\subsection{Analytical proof of convergence \label{app:convergence_of_modified_midpoint}}

This section proves analytically that the modified midpoint methods (\ref{eq:modified_midpoint_intra}), (\ref{eq:modified_midpoint_inter}), and (\ref{eq:modified_midpoint_inter}) converge to SDEs~(\ref{eq:intra_particle_sys_stratonovich})-(\ref{eq:inter_particle_sys_stratonovich}). The proof is based on the fundamental theorem of convergence for numerical approximation for SDE discussed in Ref.~\cite{milstein2004stochastic}. However, unfortunately, the fundamental theorem requires the coefficients in SDE to be smooth enough (globally Lipschitz), which is not satisfied by our SDEs (\ref{eq:intra_particle_sys_stratonovich})-(\ref{eq:inter_particle_sys_stratonovich}) when $\gamma<0$. Therefore, our proof is rigorous when $\gamma\geq 0$, but can only be treated as a heuristic argument when $\gamma<0$. Here, again, we only present the proof for the inter-species collision.

\subsubsection{Fundamental theorem and proof strategy}

Here, we briefly repeat the fundamental theorem of convergence. Firstly, we define two types of error between stochastic processes: the strong error and the weak error. Let $X_t$ and $\bar{X}_t$ be two random processes, $\langle\cdot\rangle$ be the expectation, we have
\begin{subequations}
	\begin{align}
		\text{Strong error:}\quad &
		\epsilon_\mathrm{s} \doteq 
		\left\langle|X_t - \bar{X}_t|^2\right\rangle^{1/2},
		\\ 
		\text{Weak error:} \quad &
		\epsilon_\mathrm{w} \doteq 
		|\langle X_t - \bar{X}_t \rangle|.
	\end{align}
\end{subequations}
Intuitively, the strong error is the expectation of the absolute difference, while the weak error is the absolute difference of two expectations. 

Let $X_t$ be the exact solution to an SDE, while $\bar{X}_t$ is a numerical approximation. The fundamental theorem describes the relationship between the local error of $\bar{X}_t$ in one step and the global convergence of $\bar{X}_t$. Let $\bar{X}_t(t+\Delta t)$ be a one-step approximation with time step $\Delta t$ and exact initial condition $\bar{X}_t(t) = X_t(t)$, the fundamental theorem adapted from \cite{milstein2004stochastic} is:
\begin{theorem}
Suppose the weak and strong error between the one-step approximation $\bar{X}_t(t+\Delta t)$ and the exact solution $X_t(t+\Delta t)$ are of order $p_1$ and $p_2$, i.e., $\epsilon_\mathrm{w} \sim \mathcal{O}(\Delta t^{p_1})$ and $\epsilon_\mathrm{s} \sim \mathcal{O}(\Delta t^{p_2})$. Also, let $p_2 \geq \frac{1}{2}$, $p_1 \geq p_2 + \frac{1}{2}$. Then, for all $t_k \in [0, T]$, the approximation $\bar{X}_t(t_k)$ has a strong error of order $p_2 - \frac{1}{2}$ compared to the exact solution $X_t(t_k)$. 
\end{theorem}

In order to utilize the fundamental theorem, we can compare our modified midpoint method~(\ref{eq:modified_midpoint_inter}), (\ref{eq:modified_midpoint_inter}) with other methods with known convergence rate, which we choose the Euler-Maruyama method for the It\^{o} SDE~(\ref{eq:inter_particle_sys_1}), (\ref{eq:inter_particle_sys_2}). The Euler-Maruyama method can be written as:
\begin{widetext}
	\begin{subequations}
		\label{eq:euler_maruyama}
		\begin{align}
			\tilde{\mathbf{v}}^{\alpha,i}_{k+1} - \mathbf{v}^{\alpha,i}_k
			= &
			\sum_{j=1}^{N_\beta}
			\left[
				\dfrac{ wL_{\alpha\beta}}{2 m_{\alpha} m_\mathrm{r}} 
				b(\mathbf{u}^{ij}_k) \Delta t
				+ \dfrac{\sqrt{w L_{\alpha\beta}}}{m_\alpha} 
				\boldsymbol{\Omega}^{ij}_k \times \mathbf{u}_{k}^{ij} 
			\right],
			\label{eq:euler_maruyama_1} \\ 
			\tilde{\mathbf{v}}^{\beta,j}_{k+1} - \mathbf{v}^{\beta,j}_{k}
			= & - \sum_{i=1}^{N_\alpha}
			\left[ 
				\dfrac{ wL_{\alpha\beta}}{2 m_{\beta} m_\mathrm{r}}
				b(\mathbf{u}^{ij}_k) \Delta t
				+ \dfrac{\sqrt{w L_{\alpha\beta}}}{m_\beta} 
				\boldsymbol{\Omega}^{ij}_k \times \mathbf{u}_{k}^{ij}
			\right]. 
			\label{eq:euler_maruyama_2}
		\end{align}
	\end{subequations}
\end{widetext}
To distinguish the two numerical methods, we use $\tilde{\mathbf{v}}_{k+1}$ for Euler-Maruyama and $\mathbf{v}_{k+1}$ for the modified midpoint. Notice that in each step, $\mathbf{v}_k$ is treated as given constants while $\Delta \mathbf{W}$, $\tilde{\mathbf{v}}_{k+1}$, and $\mathbf{v}_{k+1}$ are random variables. 

The Euler-Maruyama method~(\ref{eq:euler_maruyama}) is known to have one-step weak and strong errors of order $3/2$ and $1$. By using the following triangular inequalities:
\begin{subequations}
	\label{eq:triangular_inequality_1}
	\begin{align}
	|X_t - \bar{X}_t| &\leq |X_t - \tilde{X}_t| + |\tilde{X}_t - \bar{X}_t|, \\
	|X_t - \bar{X}_t|^2 &\leq 2 |X_t - \tilde{X}_t|^2 + 2 |\tilde{X}_t - \bar{X}_t|^2,
	\end{align}
\end{subequations}
we only need to prove weak and strong errors between $\tilde{\mathbf{v}}_{k+1}$ and $\mathbf{v}_{k+1}$ are of order $3/2$ and $1$ in each step. 

Before moving on to the proof, we also need to mention the following inequality and identity. Let $\mathbf{X}$ and $\mathbf{Y}$ be two 3-D vectors, so we have the inequality
\begin{equation}
	|\mathbf{X} \cdot \mathbf{Y}| \leq |\mathbf{X}| \cdot |\mathbf{Y}|,
	\qquad 
	|\mathbf{X} \times \mathbf{Y}| \leq |\mathbf{X}| \cdot |\mathbf{Y}|. 
	\label{eq:triangular_inequality_2}
\end{equation}
For a 1-D Gaussian random variable $\Delta W \sim \mathcal{N}(0,\Delta t)$, we have the following expectations:
\begin{subequations}
	\label{eq:expectation_of_gaussian_varialbes}
	\begin{align}
		\langle |\Delta W| \rangle &= \sqrt{\dfrac{2 \Delta t}{\pi}}, \\
		\langle \Delta W^2 \rangle &= \Delta t, \\
		\langle  |\Delta W|^3 \rangle &= \sqrt{\dfrac{8 \Delta t^3}{\pi}}.
	\end{align}
\end{subequations}
In addition, for a 3-D Gaussian random variable $\Delta \mathbf{W} \sim \mathcal{N}(0, \mathrm{I}_3 \Delta t)$, we have:
\begin{subequations}
	\label{eq:gaussian_identities}
	\begin{align}
		\langle\Delta \mathbf{W}^2 \rangle &= \sum_{m=1}^3 \langle\Delta W_m^2\rangle = 3 \Delta t, \\
		\langle\Delta W_m \Delta W_n\rangle &= \Delta t \, \delta_{mn}, \quad m,n\in \{1,2,3\}.
	\end{align}
\end{subequations}

\subsubsection{One-step strong error}

This section calculates the one-step difference between the Euler-Maruyama method and the modified midpoint method. The strong error is relatively simple and is estimated here. The weak error is more complicated and will be estimated in the next section. Assume at step $k$, the values of all $\mathbf{v}^{s,i}_k$ are the same for the two methods. Then from Eq.~(\ref{eq:modified_midpoint_inter}) and Eq.~(\ref{eq:euler_maruyama_1}), the difference in the next step $\boldsymbol{\epsilon}^{\alpha,i}\doteq {\mathbf{v}}_{k+1}^{\alpha,i} - \tilde{\mathbf{v}}_{k+1}^{\alpha,i}$ is given by
\begin{subequations}
	\label{eq:one_step_diff_1}
	\begin{align}
		\boldsymbol{\epsilon}^{\alpha,i}
		= & \sum_{j=1}^{N_\beta}
		\Bigg[
		\dfrac{\sqrt{w L_{\alpha\beta}}}{m_\alpha} 
		\boldsymbol{\Omega}^{ij}_k
		\times 
		\left( \mathbf{u}_{k+\frac{1}{2}}^{ij} - \mathbf{u}_{k}^{ij} \right) \\ 
		& \qquad \quad - \dfrac{ wL_{\alpha\beta}}{2 m_{\alpha} m_\mathrm{r}}
		b(\mathbf{u}^{ij}_k) \Delta t
		\Bigg]. 
	\end{align}
\end{subequations}
Making use of Eq.~(\ref{eq:modified_midpoint_inter}), we get
\begin{subequations}
	\label{eq:delta_u}
	\begin{align}
		& \mathbf{u}_{k+1/2}^{ij} - \mathbf{u}_{k}^{ij} 
		= \dfrac{1}{2}\left( \mathbf{u}_{k+1}^{ij} - \mathbf{u}_{k}^{ij} \right) 
		\\
		= & \dfrac{1}{2}
		\left[ 
			(\mathbf{v}^{\alpha,i}_{k+1} - \mathbf{v}^{\alpha,i}_k) - (\mathbf{v}^{\beta,j}_{k+1} - \mathbf{v}^{\beta,j}_k) 
		\right]
		\\ 
		= & 
		\dfrac{\sqrt{w L_{\alpha\beta}}}{2}
		\Bigg(
		\dfrac{1}{m_\alpha}
		\sum_{l=1}^{N_\beta}
		\boldsymbol{\Omega}^{il}_k \times \mathbf{u}_{k+\frac{1}{2}}^{il} 
		\\ & \qquad \qquad \quad +
		\dfrac{1}{m_\beta} 
		\sum_{l=1}^{N_\alpha}
		\boldsymbol{\Omega}^{lj}_k \times \mathbf{u}_{k+\frac{1}{2}}^{lj}
		\Bigg).
	\end{align}
\end{subequations}

Using the triangular inequalities (\ref{eq:triangular_inequality_1}) and (\ref{eq:triangular_inequality_2}), we can estimate the one-step strong error based on the above two equations as
\begin{widetext}
	\begin{subequations}
		\label{eq:strong_error_estimate}
		\begin{align}
			|\boldsymbol{\epsilon}^{\alpha,i}|^2 \leq
			&
			\sum_{j=1}^{N_\beta} \left|\dfrac{ wL_{\alpha\beta}}{2 m_{\alpha} m_\mathrm{r}}
			b(\mathbf{u}^{ij}_k) \Delta t\right|^2 
			+ \dfrac{{w L_{\alpha\beta}}}{2 m_\alpha^2} \sum_{j=1}^{N_\beta}  \sum_{l=1}^{N_\beta}
			|\mathbf{u}^{ij}_k|^{\gamma} \cdot |\Delta \mathbf{W}^{ij}|^2 \cdot |\mathbf{u}^{il}_k|^{\gamma} 
			\cdot |\Delta \mathbf{W}^{il}|^2 \cdot |\mathbf{u}_{k+1/2}^{il}|^2 
			\\ 
			& + \dfrac{{w L_{\alpha\beta}}}{2 m_\alpha m_\beta} \sum_{j=1}^{N_\beta} \sum_{l=1}^{N_\alpha}
			|\mathbf{u}^{ij}_k|^{\gamma} \cdot |\Delta \mathbf{W}^{ij}|^2 \cdot |\mathbf{u}^{lj}_k|^{\gamma}
			\cdot |\Delta \mathbf{W}^{lj}|^2 \cdot |\mathbf{u}_{k+1/2}^{lj}|^2.
		\end{align}
	\end{subequations}
\end{widetext}
The term $\mathbf{u}^{ij}_{k+1/2}$ is a random variable since it involves $\mathbf{v}^{s,i}_{k+1}$. Thanks to the exact energy conservation of the modified midpoint methods, their magnitudes are always bounded by the initial energy $E_0$ of the system, i.e., $|\mathbf{v}^{s,i}_{k+1}|^2 \leq C E_0$, where $C$ is some constant related to particle mass and weight. We can show $|\mathbf{u}^{ij}_{k+1/2}|$ is bounded:
\begin{subequations}
	\label{eq:boundedness_of_uij}
	\begin{align}
		& |\mathbf{u}^{ij}_{k+1/2}|^2 = \dfrac{1}{4}
		\left|
			\mathbf{v}^{\alpha,i}_{k+1} - \mathbf{v}^{\beta,i}_{k+1} 
			+ \mathbf{v}^{\alpha,i}_{k} - \mathbf{v}^{\beta,i}_{k} 
		\right|^2 \\
		& \leq
			|\mathbf{v}^{\alpha,i}_{k+1}|^2 + |\mathbf{v}^{\beta,i}_{k+1}|^2
			+ |\mathbf{v}^{\alpha,i}_{k}|^2 + |\mathbf{v}^{\beta,i}_{k}|^2 \\ 
		& \leq 4 E_0. 
	\end{align}
\end{subequations}
So, with the boundedness of $|\mathbf{u}^{ij}_{k+1/2}|^2$, the only random variables in Eq.~(\ref{eq:strong_error_estimate}) are the Gaussian random variables $\Delta \mathbf{W}^{ij}$. Using Eq.~(\ref{eq:expectation_of_gaussian_varialbes}), we find that the expectation of the RHS in Eq.~(\ref{eq:strong_error_estimate}) are all of order $\mathcal{O}(\Delta t^2)$, which means the strong error $\sqrt{\langle|\boldsymbol{\epsilon}^{\alpha,i}|^2\rangle}$ is at most $\mathcal{O}(\Delta t)$. The strong error of $\boldsymbol{\epsilon}^{\beta,j}$ can be estimated in the same way. Therefore, the one-step strong error of the modified midpoint method is at least order 1.

\subsubsection{One-step weak error}

To prove that the weak error $|\langle\boldsymbol{\epsilon}^{\alpha,i}\rangle|$ is at most $\mathcal{O}(\Delta t^{3/2})$, we need a more careful estimation of Eq.~(\ref{eq:one_step_diff_1}) and (\ref{eq:delta_u}). In particular, the drag term in Eq.~(\ref{eq:one_step_diff_1}) is $\mathcal{O}(\Delta t)$. So, we need to prove that the drag terms are somehow canceled out by the diffusion term. The strategy is to expand the $\mathbf{u}^{ij}_{k+{1}/{2}}$-terms on the RHS of Eq.~(\ref{eq:delta_u}) once again and calculate its leading correction. As we will see later, the leading correction cancels out the drag term exactly.

Firstly, due to the boundedness of $\mathbf{u}_{k+1/2}^{ij}$, we can use Eq.~(\ref{eq:delta_u}) to get an estimation:
\begin{equation}
	\mathbf{u}_{k+1/2}^{ij} - \mathbf{u}_{k}^{ij} \sim  \mathcal{O}(\sqrt{\Delta t}\,).
\end{equation}
Plug this estimation into the RHS of Eq.~(\ref{eq:delta_u}), we get the next order estimation:
\begin{subequations}
	\label{eq:delta_u_2}
	\begin{align}
		\mathbf{u}_{k+1/2}^{ij} - \mathbf{u}_{k}^{ij} = & 
		\dfrac{\sqrt{w L_{\alpha\beta}}}{2 m_\alpha}
		\sum_{l=1}^{N_\beta}
		\boldsymbol{\Omega}^{il}_k \times \mathbf{u}_{k}^{il} \\ 
		& +
		\dfrac{\sqrt{w L_{\alpha\beta}}}{2 m_\beta} 
		\sum_{l=1}^{N_\alpha}
		\boldsymbol{\Omega}^{lj}_k \times \mathbf{u}_{k}^{lj} \\*
		& + \mathcal{O}(\Delta t),
	\end{align}
\end{subequations}
where we have used $\boldsymbol{\Omega} \sim \Delta \mathbf{W} \sim \mathcal{O}(\sqrt{\Delta t}\,)$. Using Eq.~(\ref{eq:delta_u_2}) in the one-step difference~(\ref{eq:one_step_diff_1}), we get
\begin{subequations}
	\begin{align}
		\boldsymbol{\epsilon}^{\alpha,i} = & 
		\dfrac{w L_{\alpha\beta}}{2 m_\alpha} \sum_{j=1}^{N_\beta}
		\Bigg[
		\boldsymbol{\Omega}^{ij}_k 
		\times 
		\left(
			\dfrac{1}{m_\alpha} \sum_{l=1}^{N_\beta}
			\boldsymbol{\Omega}^{il}_k \times \mathbf{u}_{k}^{il}
		\right) 
		\label{eq:one_step_diff_2} \\ 
		& \qquad\qquad + 
		\boldsymbol{\Omega}^{ij}_k 
		\times 
		\left(
			\dfrac{1}{m_\beta} \sum_{l=1}^{N_\alpha}
			\boldsymbol{\Omega}^{lj}_k  \times \mathbf{u}_{k}^{lj}
		\right)
		\label{eq:one_step_diff_3} \\ 
		& \qquad\qquad - \dfrac{1}{m_\mathrm{r}}
		b(\mathbf{u}^{ij}_k) \Delta t 
		\Bigg]
		+ \mathcal{O}\big(\Delta t^{3/2}\big).
		\label{eq:one_step_diff_4}
	\end{align}
\end{subequations}

Notice that when $l\neq j$, $\Delta \mathbf{W}^{ij}$ and $\Delta \mathbf{W}^{il}$ are independent Gaussian variables with mean 0, so their expectation vanishes. And so is $\Delta \mathbf{W}^{ij}$ and $\Delta \mathbf{W}^{lj}$. Therefore, when we take the expectation $\langle\boldsymbol{\epsilon}^{\alpha,i} \rangle$, all terms in the summation $\sum_l$ \yfu{vanish} except the terms $l=j$ or $l=i$. Let the first two lines~(\ref{eq:one_step_diff_2}) and (\ref{eq:one_step_diff_3}) within the bracket be $\mathbf{T}^{ij}$, whose expectation is 
\begin{widetext}
	\begin{subequations}
		\begin{align}
			\langle \mathbf{T}^{ij} \rangle \doteq 
			& \left\langle
			\boldsymbol{\Omega}^{ij}_k \times 
			\left(
				\dfrac{1}{m_\alpha} \sum_{l=1}^{N_\beta} \boldsymbol{\Omega}^{il}_k \times \mathbf{u}_{k}^{il}
			\right)
			+ 
			\boldsymbol{\Omega}^{ij}_k \times 
			\left(
				\dfrac{1}{m_\beta} \sum_{l=1}^{N_\alpha} \boldsymbol{\Omega}^{lj}_k \times \mathbf{u}_{k}^{lj}
			\right)
			\right\rangle \\ 
			= & \left\langle
				\boldsymbol{\Omega}^{ij}_k \times 
				\left(
					\dfrac{1}{m_\alpha} \boldsymbol{\Omega}^{ij}_k \times \mathbf{u}_{k}^{ij}
				\right)
				+ 
				\boldsymbol{\Omega}^{ij}_k \times 
				\left(
					\dfrac{1}{m_\beta} \boldsymbol{\Omega}^{ij}_k \times \mathbf{u}_{k}^{ij}
				\right)
			\right\rangle \\ 
			= & \dfrac{1}{m_\mathrm{r}}
			\left\langle
				\boldsymbol{\Omega}^{ij}_k \times 
				\left(
					\boldsymbol{\Omega}^{ij}_k \times \mathbf{u}_{k}^{ij}
				\right)
			\right\rangle \\ 
			= & \dfrac{1}{m_\mathrm{r}}
			|\mathbf{u}^{ij}_k|^{\gamma-2}
			\left\langle
				\big| \mathbf{u}^{ij}_k \cdot \Delta \mathbf{W}^{ij}\big|^2 -
				\big|\mathbf{u}^{ij}_k\big|^2 \cdot \big|\Delta \mathbf{W}^{ij}\big|^2
			\right\rangle \mathbf{u}^{ij}_k.
		\end{align}
	\end{subequations}
\end{widetext}
For any constant vector $\mathbf{A}=(A_1,A_2,A_3)^\intercal$ and 3-D Gaussian random variable $\Delta \mathbf{W}\sim\mathcal{O}(0,\mathrm{I}_3 \Delta t)$, we can use the identities (\ref{eq:gaussian_identities}) to get
\begin{subequations}
	\begin{align}
		\left\langle|\mathbf{A}\cdot \Delta \mathbf{W}|^2 \right\rangle 
		& = \sum_{n=1}^3 \langle A_n^2 \Delta W_n^2 \rangle + \text{cross terms} \\ 
		& = |\mathbf{A}|^2 \Delta t.
	\end{align}
\end{subequations}
Thus, the first line~(\ref{eq:one_step_diff_2}) within the bracket can be simplified as 
\begin{align}
	\left\langle\mathbf{T}^{ij}\right\rangle 
	= & \dfrac{1}{m_\mathrm{r}}
	|\mathbf{u}^{ij}_k|^{\gamma-2}
	\left(
		\big| \mathbf{u}^{ij}_k \big|^2 \Delta t - 
		\big|\mathbf{u}^{ij}_k\big|^2 3 \Delta t
	\right) 
	\mathbf{u}^{ij}_k 
	\\ 
	= & \dfrac{1}{m_\mathrm{r}}
	\left(
		-2 |\mathbf{u}^{ij}_k|^{\gamma} \mathbf{u}^{ij}_k 
	\right) \Delta t
	= \dfrac{1}{m_\mathrm{r}} b(\mathbf{u}^{ij}_k) \Delta t,
\end{align}
which cancels out the third line~(\ref{eq:one_step_diff_4}) in the bracket exactly. So, the weak error $|\langle\boldsymbol{\epsilon}^{\alpha,i}\rangle|$ is at most $\mathcal{O}(\Delta t^{3/2})$. The estimation of the weak error of $\boldsymbol{\epsilon}^{\beta,j}$ yields the same result. 

Therefore, combining the strong and weak error, the fundamental theorem proves that the modified midpoint method converges to SDE~(\ref{eq:inter_particle_sys}) with a strong error of order $1/2$ in the whole time $t_k\in[0,T]$.

\subsection{Numerical demonstration of convergence rate}

This section presents a numerical verification of the convergence rate of the modified midpoint method. The method of calculating strong error numerically is the same as Refs.~{\cite{zhang2020simulating,fu2022explicitly}}. In a fixed duration of time period $t\in[0,T]$, one can define a series of time steps $\Delta t^l \doteq T/2^l$ for $l\in\{0,1,2,...\}$. With each time step $\Delta t^l$, the time is discretized into $t_k^l \in \{t_0^l, t_1^l, t_2^l, ...\}$ where $t_{k+1}^l = t_k^l + \Delta t^l$. For one given Wiener process $W(t)$, it can be discretized using a series of time steps as 
\begin{equation}
	\Delta W^l_k \doteq W(t^l_{k+1}) - W(t^l_{k}) \sim \mathcal{N}(0, \Delta t^l).
\end{equation}
Such a discretization method guarantees the $\Delta W^l_k$ in the series of time steps corresponding to the same Wiener process so that they are correlated:
\begin{equation}
	\Delta W^l_k = \Delta W^{l+1}_{2k} + \Delta W^{l+1}_{2k+1}.
\end{equation}
Using $\Delta W^l_k$ at multiple different levels $l$, one can compare the difference of multiple numerical simulations since they correspond to the same underlying Wiener process. 

We calculate sample paths of SDEs~(\ref{eq:intra_particle_sys_stratonovich})-(\ref{eq:inter_particle_sys_stratonovich_2}) using the modified midpoint method with multiple levels of time step $\Delta t^l$. The cases we tested are the same as the numerical examples in the main article with the same physical parameters. However, we slightly change the numerical parameters to reduce the computational complexity. For intra-species temperature isotropization, we use $N=32$ particles and choose the time period $T=10^{-2} \tau_\mathrm{iso,0}$. One series of sample paths of the $v_x$ component of the first particle is presented in Fig.~\ref{fig:strong_convergence}(a). We can see the sample path gradually converges when the time step becomes smaller.

\begin{figure}[ht]
	\centering 
	\includegraphics[width=0.7\columnwidth]{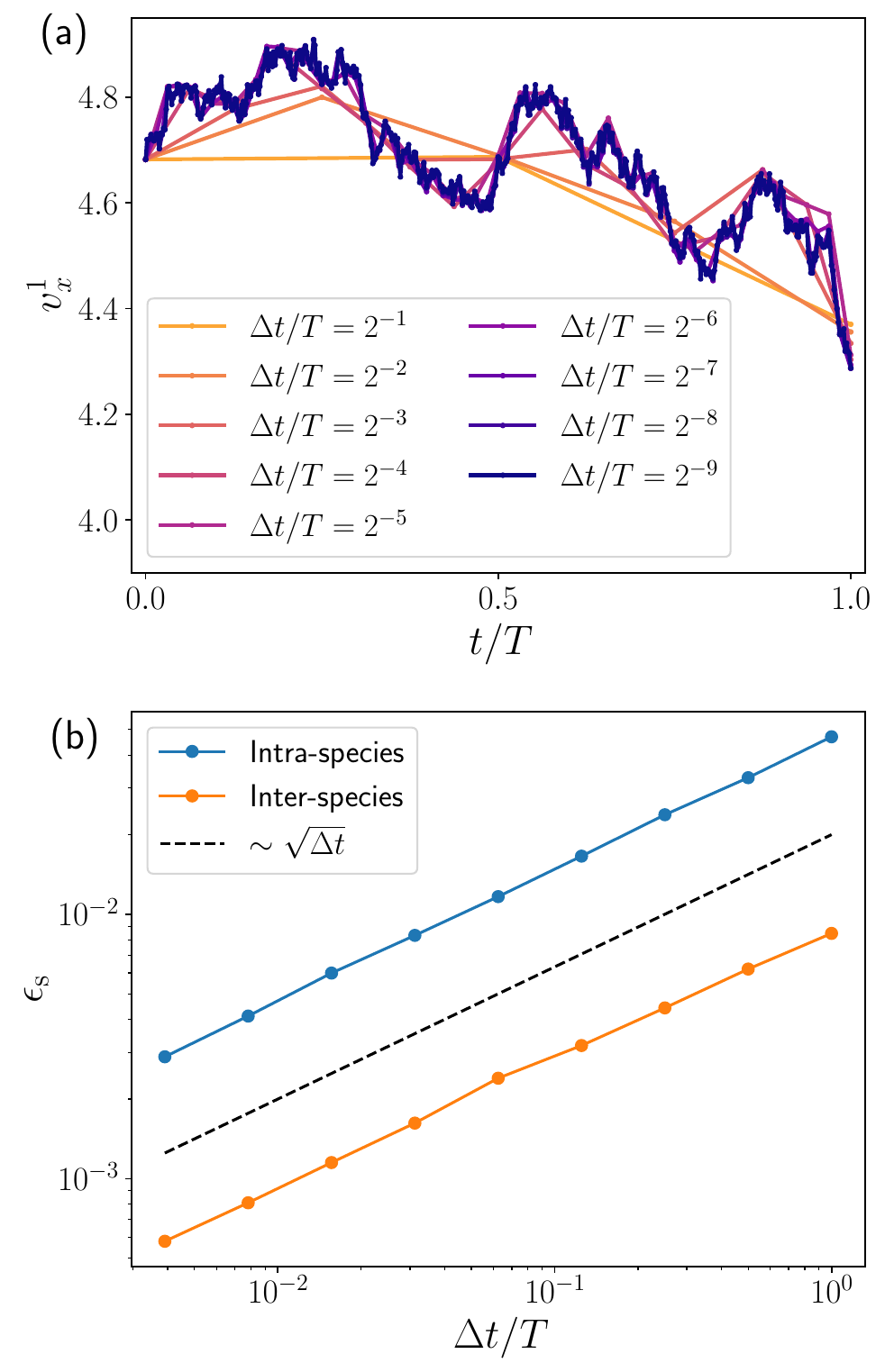} 
	\caption{The strong convergence of the modified midpoint method. (a) the convergence of each sample path with different time steps. (b) the strong convergence rate.}
	\label{fig:strong_convergence}
\end{figure}

To calculate the strong error, we only need to look at the value of the sample path at the last time step. Since the exact is unknown, we calculate the relative error between time steps $\Delta t^l$ and $\Delta t^{l+1}$ as
\begin{equation}
	\epsilon_\mathrm{s}(\Delta t^l) \doteq 
	\left\langle 
		\left| \mathbf{v}(T;\Delta t^{l+1}) - \mathbf{v}(T;\Delta t^{l}) \right|^2
	\right\rangle^{1/2}.
\end{equation}
As shown in Ref.~\cite{fu2022explicitly}, if the algorithm has strong convergence of order $p$, the relative error should have the same convergence order. As numerical examples, we calculated the strong error for both intra- and inter-species collision. For intra-species temperature isotropization, 32 particles are used with time period $T=5 \times 10^{-4} \tau_\mathrm{iso,0}$. For inter-species temperature relaxation, (16, 32) particles are used for two species with time period $T=10^{-4} \tau_\mathrm{12,0}$. The strong errors calculated over 128 independent ensembles are shown in Fig.~\ref{fig:strong_convergence}(b). We can see the strong convergence rate for both species is of order $1/2$. It is worth mentioning that if the time period is too large, the convergence rate would be slower than $1/2$.

\yfu{
\subsection{Complexity reduction \label{sec:complexity_reduction}}

In algorithm~(\ref{eq:modified_midpoint_intra}), each particle interacts with all other particles, so the total cost is at least $\mathcal{O}(N^2)$. Explicitly solving Eq.~(\ref{eq:modified_midpoint_intra}) also requires the expensive inversion of large matrices, see Appendix~\ref{app:explicit_solution_to_modified_midpoint}. To reduce complexity, we can use the technique introduced by Takizuka and Abe \cite{takizuka1977binary}. Instead of counting the interaction among all particles, we can randomly separate particles into groups, and calculate the collision only within each group. 

For example in inter-species collision, $N_\alpha$ and $N_\beta$ particles are randomly divided into $N_\mathrm{g}$ groups of ${\tilde{N}}_\alpha \doteq N_\alpha/N_\mathrm{g}$ and $\tilde{N}_\beta \doteq N_\beta/N_\mathrm{g}$ particles at each time step. Within each group, the distribution function of each species in Eq.~(\ref{eq:particle_distribution}) can be represented by only ${\tilde{N}}_\alpha$ and ${\tilde{N}}_\beta$ particles with the weight of particles increased to $\tilde{w} = w N_\mathrm{g}$. To exclude the self-interaction, the particle weight intra-species collision should be $\tilde{w}=n / (N/N_\mathrm{g}-1)$. This reduced distribution is used to calculate the collision using Eq.~(\ref{eq:modified_midpoint_intra}). In the next time step, particles are divided into different random groups, and the above procedure is repeated.

The grouping procedure decreases the number of interactions each particle encounters and reduces computational complexity. In the case where two species have the same density $n$ (so same particle number $N$), maximum grouping with $N_\mathrm{g}=N$ is equivalent to binary collisions. Then, Eq.~(\ref{eq:modified_midpoint_intra}) represents a pure rotation in the center-of-mass (CM) frame in each binary pair. Let the deflection angle between particle $i$ and $j$ after one step in the CM frame be $\theta$. As shown in Appendix~\ref{app:grouping}, the distribution of $\delta \doteq \tan(\theta/2)$ has variance
\begin{equation}
	\langle \delta^2 \rangle = 
	\left(\dfrac{e_\alpha^2 e_\beta^2 n \ln\Lambda_{\alpha\beta}}{8\pi \epsilon_0^2 m_\mathrm{r}^2 |\mathbf{u}^{ij}_k|^3} \right)\Delta t 
	+ \mathcal{O}(\Delta t^{3/2}),
\end{equation}
where $m_\mathrm{r}=m_\alpha m_\beta / (m_\alpha+m_\beta)$ is the reduced mass of two particles. Therefore, in this limit, our method returns to the pair collision algorithm in Ref.~\cite{takizuka1977binary}.

\begin{figure}[ht]
	\centering 
	\includegraphics[width=0.95\columnwidth]{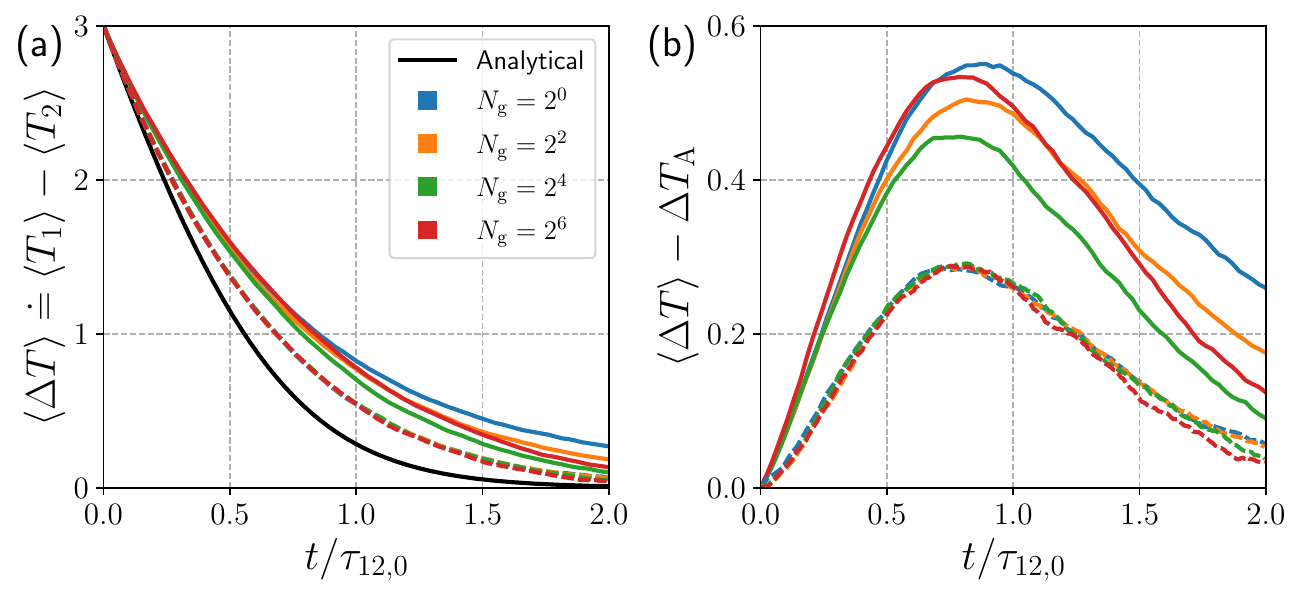} 
	\caption{Effect of particle grouping. (a) Average temperature difference $\langle \Delta T \rangle$ with various grouping $N_\mathrm{g}$. Solid and dashed lines represent $\Delta t / \tau_{11,0} = 10^{-1}, 10^{-3}$, respectively. (b) The differences between $\langle \Delta T \rangle$ and analytical solution $\Delta T_\mathrm{A}$.}
	\label{fig:particle_grouping}
\end{figure}

The particle grouping is demonstrated in the inter-species temperature relaxation case with two different time steps $\Delta t /  \tau_{11,0} = 10^{-1}, 10^{-3}$. Fig.~\ref{fig:particle_grouping} shows the average temperature difference $\langle \Delta T \rangle$ with different numbers of groups $N_\mathrm{g}$. When $\Delta t$ is large, counter-intuitively, $\langle \Delta T \rangle$ becomes closer to the analytical solution when $N_\mathrm{g}$ increases from $1$ to $2^4$. It is speculated that this occurs because the coefficient $\sigma(\mathbf{u})\to\infty$ when $|\mathbf{u}|\to0$. Heuristically, the larger $\sigma(\mathbf{u})$ is, the smaller $\Delta t$ should be to ensure numerical accuracy. When one large $\sigma(\mathbf{u})$ occurs, the number of particles affected in one step decreases as $N_\mathrm{g}$ increases, which reduces the overall error such large $\sigma(\mathbf{u})$ introduces. When $N_\mathrm{g}$ increases further from $2^4$ to $2^6$, the numerical error increases, likely due to the over-simplification of distribution functions. When $\Delta t$ is small, the accuracy among different $N_\mathrm{g}$ becomes negligible. This result suggests that the grouping technique may reduce computational complexity and increase accuracy when the time step is large for some optimal number of groups.

}

\subsection{Distribution of scattering angle in the pair-collision limit \label{app:grouping}}

In this final section, we discuss the distribution of scattering angles in the pair-collision limit when the particle grouping technique is applied. 

For inter-species collision where equal density $n$, the number of particles in each species should be the same as $N$. Therefore, if we divide the particles into $N$ groups, then the collision within each group becomes pair collision, where the weight of each particle becomes $w=n$. In this case, the modified midpoint algorithm within each group becomes
\begin{subequations}
	\begin{align}
		\mathbf{v}^{\alpha}_{k+1} - \mathbf{v}^{\alpha}_k
		&= \dfrac{\sqrt{n L_{\alpha\beta}}}{m_\alpha}
		\left( 
			\dfrac{\mathbf{u}_k \times \Delta \mathbf{W} }{|\mathbf{u}_k|^{1-\gamma/2}}
		\right) \times \mathbf{u}_{k+\frac{1}{2}}, 
		\label{eq:modified_midpoint_inter_pair1}\\
		\mathbf{v}^{\beta}_{k+1} - \mathbf{v}^{\beta}_{k}
		&= - \dfrac{\sqrt{n L_{\alpha\beta}}}{m_\beta} 
		\left( 
			\dfrac{\mathbf{u}_k \times \Delta \mathbf{W} }{|\mathbf{u}_k|^{1-\gamma/2}}
		\right) \times \mathbf{u}_{k+\frac{1}{2}},
		\label{eq:modified_midpoint_inter_pair2}
	\end{align}
\end{subequations}
where the particle index $i$ and $j$ are omitted because there is only one particle in each species within each group, and $\mathbf{u}_k \doteq \mathbf{v}^\alpha_k - \mathbf{v}^\beta_k$.

The difference (\ref{eq:modified_midpoint_inter_pair1})-(\ref{eq:modified_midpoint_inter_pair2}) gives us the evolution of the relative speed $\mathbf{u}$, which is
\begin{equation}
	\mathbf{u}_{k+1} - \mathbf{u}_k
	= \dfrac{\sqrt{n L_{\alpha\beta}}}{m_\mathrm{r}}
	\left( 
		\dfrac{\mathbf{u}_k \times \Delta \mathbf{W} }{|\mathbf{u}_k|^{1-\gamma/2}}
	\right) \times \mathbf{u}_{k+\frac{1}{2}}.
\end{equation}
Multiply the equation above by $\mathbf{u}_{k+1} + \mathbf{u}_k \equiv 2 \mathbf{u}_{k+1/2}$ on both sides, we get $|\mathbf{u}_{k+1}|^2 = |\mathbf{u}_{k}|^2$. So the equation describes a pure rotation of the relative velocity $\mathbf{u}_k$. Let the time step be $\Delta t$, the angle between $\mathbf{u}_k$ and $\mathbf{u}_{k+1}$ be $\theta$ can be calculated as:
\begin{subequations}
	\begin{align}
		& \sin(\theta/2) =
		\dfrac{|\mathbf{u}_{k+1} - \mathbf{u}_k|}{2|\mathbf{u}_k|} \\[2pt]
		= & \dfrac{\sqrt{n L_{\alpha\beta}}}{2 m_\mathrm{r}}
		\dfrac{\left|(\mathbf{u}_k \times \Delta \mathbf{W})\times \mathbf{u}_{k+1/2}\right|}{|\mathbf{u}_k|^{2-\gamma/2}}
		\\[2pt]
		\approx & 
		\dfrac{\sqrt{n L_{\alpha\beta}}}{2  m_\mathrm{r}}
		\dfrac{\left|(\mathbf{u}_k \times \Delta \mathbf{W})\times \mathbf{u}_{k}\right|}{|\mathbf{u}_k|^{2-\gamma/2}}
		+ \mathcal{O}({\Delta t}),
	\end{align}
\end{subequations}
where we have used $\mathbf{u}_{k+1/2} - \mathbf{u}_{k} \sim \Delta \mathbf{W} \sim \mathcal{O}(\sqrt{\Delta t}\,)$.

The variance of $\sin(\theta/2)$ can also be calculated as:
\begin{subequations}
	\begin{align}
		\left\langle [\sin(\theta/2)]^2 \right\rangle
		= & \dfrac{n L_{\alpha\beta}}{4 m_\mathrm{r}^2}
		\dfrac{\left\langle\left|(\mathbf{u}_k \times \Delta \mathbf{W})\times \mathbf{u}_{k}\right|^2\right\rangle}{|\mathbf{u}_k|^{4-\gamma}} \\ 
		& + \mathcal{O}(\Delta t^{3/2}).
	\end{align}
\end{subequations}
The expectation above is
\begin{subequations}
	\begin{align}
		& \left\langle\left|(\mathbf{u}_k \times \Delta \mathbf{W})\times \mathbf{u}_{k}\right|^2\right\rangle \\ 
		= & \left\langle
			\left|
			|\mathbf{u}_k|^2 \Delta \mathbf{W}
			- (\mathbf{u}_k \cdot \Delta \mathbf{W}) \mathbf{u}_k
			\right|^2
		\right\rangle \\
		= & |\mathbf{u}_k|^2
		\left\langle
			|\mathbf{u}_k|^2 |\Delta \mathbf{W}|^2
			- |\mathbf{u}_k \cdot \Delta \mathbf{W}|^2
		\right\rangle \\ 
		= & 2 |\mathbf{u}_k|^4 \Delta t.
	\end{align}
\end{subequations}
So we get
\begin{align}
	& \left\langle [\sin(\theta/2)]^2 \right\rangle = 
	\dfrac{n L_{\alpha\beta}}{4 m_\mathrm{r}^2} 
	2 |\mathbf{u}_k|^\gamma \Delta t
	+ \mathcal{O}(\Delta t^{3/2}) \\ 
	= & \left(\dfrac{e_\alpha^2 e_\beta^2 n \ln\Lambda_{\alpha\beta}}{8\pi \epsilon_0^2 m_\mathrm{r}^2 |\mathbf{u}^{ij}_k|^{-\gamma}} \right)\Delta t 
	+ \mathcal{O}(\Delta t^{3/2}).
\end{align}
Therefore, it is easy to show that $\delta\doteq \tan(\theta/2)$ has the same variance in the leading order, which returns to the result in Takizuka and Abe \cite{takizuka1977binary}.

\bibliography{Landau.bbl}

\end{document}